\documentclass{theoretics}
\DeclareRedundantLanguages{en}{english,american,british,
  canadian,australian,newzealand,USenglish,UKenglish} 
\usepackage{dsfont}
\usepackage{comment}
\usepackage{mathtools}
\usepackage {tikz}
\usepackage{systeme}
\usepackage{colortbl}
\usepackage{nicefrac}
\usepackage{multicol}

\SetCommentSty{mycommfont}
\usepackage{xspace}

\DeclareMathOperator*{\argmin}{arg\,min}
\DeclareMathOperator{\cost}{\mathsf{cost}}

\DeclareMathOperator{\ED}{ED}
\DeclareMathOperator{\UD}{UD}

\DeclareMathOperator{\poly}{poly}
\DeclareMathOperator{\supp}{supp}
\DeclareMathOperator{\ed}{ed}
\newcommand{\R}{\mathbb{R}}
\newcommand{\Z}{\mathbb{Z}}

\DeclareMathOperator{\sgn}{sgn}

\newcommand{\norm}[1]{\left\lVert#1\right\rVert}

\newcommand{\cst}{$\mathsf{CST}$\xspace}
\newcommand{\setp}{$\mathsf{SP}$\xspace}
\newcommand{\setpT}{$\mathsf{SP_3}$\xspace}

\newcommand{\setpB}{$\mathsf{SP_B}$\xspace}
\newcommand{\dst}{$\mathsf{DST}$\xspace}

\newcommand{\np}{$\mathsf{NP}$\xspace}
\newcommand{\apx}{$\mathsf{APX}$\xspace}

\newcommand{\ptas}{$\mathsf{PTAS}$\xspace}
\newcommand{\vcov}{$\mathsf{VC}$\xspace}

\newcommand{\td}{\tau}
\newcommand{\pcoord}{\theta_p}
\newcommand{\e}{\mathbf{e}}
\newcommand{\srd}{\alpha_{\text{X}}}
\newcommand{\trd}{\alpha_{\text{P}}}
\newcommand{\reals}{\mathbb{R}_{\ge 1}\cup\{0,\infty\}}

\renewcommand{\tilde}{\widetilde}

\declaretheorem[style=italicized,name=Open Question]{openq}

\title{On Approximability of Steiner Tree in \texorpdfstring{$\ell_p$}{lp}-metrics}
\ThCSauthor[CarnegieMellon]{Henry Fleischmann}{hfleischmann3@gmail.com}[0000-0002-6093-3393]
\ThCSauthor[Rutgers]{Surya Teja Gavva}{suryateja@math.rutgers.edu}[0000-0003-0845-5029]
\ThCSauthor[Rutgers]{Karthik C.\ S.}{karthik0112358@gmail.com}[0000-0001-9105-364X]
\ThCSaffil[CarnegieMellon]{Carnegie Mellon University, USA}
\ThCSaffil[Rutgers]{Rutgers University, USA}
\ThCSthanks{Henry Fleischmann's work was supported by NSF grant CNS-2150186 and an NSF Graduate Research Fellowship under Grant No. 2140739. Part of this research was conducted while he was affiliated with the University of Michigan. The work of Karthik C.\ S.\ was supported by  Rutgers University’s Research Council Individual Fulcrum Award (\#AWD00010234),  by the National Science Foundation under Grant CCF-2313372, and a grant from the Simons
 Foundation, Grant Number 825876, Awardee Thu D. Nguyen. A preliminary version of this article appeared at SODA 2024 \cite{FleischmannGS24}.}
\ThCSshortnames{H.\ Fleischmann, S.\ T.\ Gavva, Karthik C.\ S.}  
\ThCSshorttitle{On Approximability of Steiner Tree in \texorpdfstring{$\ell_p$}{lp}-metrics}
\ThCSyear{2025}
\ThCSarticlenum{4}
\ThCSreceived{Oct 11, 2023}
\ThCSrevised{Sep 12, 2024}
\ThCSaccepted{Dec 3, 2024}
\ThCSpublished{Jan 20, 2025}
\ThCSdoicreatedtrue

\ThCSkeywords{Hardness of Approximation, Steiner Tree, Graph Coloring, Set Cover}

\begin{document}
\maketitle

\begin{abstract}
In the Continuous Steiner Tree problem (\cst), we are given as input a set of points (called \emph{terminals}) in a metric space and asked for the minimum-cost tree connecting them. Additional points (called \emph{Steiner points}) from the metric space can be introduced as nodes in the solution. In the Discrete Steiner Tree problem (\dst), we are given in addition to the terminals, a set of facilities, and any solution tree connecting the terminals can only contain the Steiner points from this set of facilities. 

Trevisan [SICOMP'00] showed that \cst and \dst are \apx-hard when the input lies in the $\ell_1$-metric (and Hamming metric). Chleb\'ik and Chleb\'ikov\'a [TCS'08] showed that \dst is \np-hard to approximate to factor of $96/95\approx 1.01$ in the graph metric (and consequently $\ell_\infty$-metric). Prior to this work, it was unclear if \cst and \dst are \apx-hard in essentially every other popular metric. 

In this work, we prove that \dst is \apx-hard in every $\ell_p$-metric. We also prove that \cst is \apx-hard in the $\ell_{\infty}$-metric. Finally, we relate \cst and \dst, by observing a gap preserving reduction from \cst to \dst in $\ell_p$-metrics. 

 It is known that the \apx-hardness of \dst in $\ell_0,\ell_1$, and $\ell_\infty$-metrics  can be obtained from the {\apx-hardness} of covering problems (with additional structure). 
Our main conceptual insight is that for certain ranges of $p$ (such as $p=2$), the soundness guarantees of covering problems might be insufficient to show that  \dst in the $\ell_p$-metric is \apx-hard, but the soundness guarantees of a packing problem (with requisite additional structure) is enough. Equipped with this insight, we are then able to embed set systems into every $\ell_p$-metric, and, depending on the value of $p$, our soundness analysis of the corresponding \dst instance in the $\ell_p$-metric will use either the packing property or the covering property (or both) of the set system in the soundness case.  

Due to the discrete structure of the Hamming space, the \apx-hardness of \cst in the $\ell_0$ and $\ell_1$ metrics follow from similar methods. However, these techniques do not extend to \cst in other $\ell_p$-metrics. To show \apx-hardness of \cst in the $\ell_{\infty}$-metric, we instead rely on the robust hardness guarantees of graph coloring problems. Concretely, we present a reduction from a graph $G$ on $n$ vertices to a point-set $P$  in the $\ell_\infty$-metric space, where the cost of the optimal Steiner tree of $P$ is exactly $\frac{n+\chi(G)}{2}$, where $\chi(G)$ is the chromatic number of $G$. 
\end{abstract}
 

\section{Introduction}

Given a set of points (called \emph{terminals}) in a metric space, a \emph{Steiner tree} is defined to be a minimum length tree connecting them, with the tree possibly including additional points (called \emph{Steiner} points) as nodes. Computing the Steiner
tree is one of the classic and most fundamental problems in Computer Science, Combinatorial Optimization, and Operations Research, with both great theoretical and practical relevance and interest \cite{ljubic2021solving}. 
This problem
emerges in a number of contexts, such as network design problems, design of integrated circuits, location problems \cite{cheng2013steiner,cho2001steiner,hwang1992steiner,lengauer2012combinatorial,resende2008handbook,noormohammadpour2017dccast}  and
more recently even in machine learning, systems biology, and bioinformatics \cite{backes2012integer,ideker2002discovering,russakovsky2010steiner, tuncbag2016network}.
For example, in VLSI circuits, wire routing is carried out by wires running only in vertical and horizontal directions, due to high computational complexity of the task. Therefore, the wire length is the sum of the lengths of vertical and horizontal segments, and the distance between two pins of a net is actually the $\ell_1$-metric  distance between the corresponding geometric points in the design plane \cite{sherwani2012algorithms}. 

Formally, the Steiner tree problem can be formulated in two ways: discrete and continuous. In the Discrete Steiner Tree problem (\dst), we are given as input two sets of points in a metric space, called \emph{terminals} and \emph{facilities} respectively. The goal is to find the minimum-cost tree connecting the terminals, possibly introducing new points (\emph{Steiner} points) from the set of facilities as nodes in the solution. This problem is well-defined even in general metric spaces since the relevant metric space can be fully described as part of the input. In the Continuous Steiner Tree problem (\cst), we are only given a set of terminals as input, and we are allowed to use any point in the metric space as a Steiner point in the Steiner tree connecting all the terminals. For \cst, we assume knowledge of the ambient metric space structure (e.g., ($(\{0,1\}^n, \ell_1)$ or $(\R^n, \ell_2)$). This is not part of the input.

\paragraph{General Metrics.}
  In general metrics, where the metric space is specified as part of the input (for example, by specifying all pairwise distances), it does not make sense to discuss the \cst problem, so we restrict our attention to \dst. The Steiner tree problem is one of Karp's 21 \np-complete problems, i.e., in his seminal work \cite{karp1972reducibility},  Karp showed that \dst is \np-complete. Therefore, the attention of the algorithmic community turned towards obtaining good approximation algorithms for \dst in general metrics. It is well-known that a minimum-cost
tree only containing the terminals as nodes is a 2-approximation for \dst (and \cst) \cite{Gilbert_Pollak_1968,Vazirani_2001}.

A sequence of improved approximation algorithms for \dst in general metrics appeared in the literature \cite{zelikovsky199311,karpinski1997new,promel2000new,robins2005tighter}, culminating with the famous $\ln(4)+\varepsilon<1.39$ factor approximation algorithm by Byrka, Grandoni, Rothvo{\ss}, and Sanit\'a \cite{byrka2013steiner} (where $\varepsilon>0$ is an arbitrarily small constant). On the hardness of approximation front, Chleb\'ik and Chleb\'ikov\'a \cite{Chlebik_Chlebikova_2008} showed that \dst in general metrics is hard to approximate to a factor of $\frac{96}{95}>1.01$. Thus, there is still quite a substantial gap in our understanding of the approximability of \dst in general metrics.

\begin{openq}What is the tight inapproximability ratio of \dst in general metrics? \label{open1}\end{openq}

\paragraph{Geometric Metrics.}
Special cases of \cst in Euclidean metric were studied by Fermat, Torricelli and other mathematicians
as early as the $17\textsuperscript{th} $ 
century and were also discussed in a letter from  Gauss to Schumacher in 1836 (see \cite{brazil2014history} for more details on the history of this problem). 
On the other hand, phylogeny reconstruction  is a long-standing (dating from Darwin’s evolutionary theory) and
intensively studied problem in computational biology, and computing the phylogenetic
tree can be modeled as solving \cst in the Hamming metric \cite{foulds1982steiner}. We have already discussed the importance of \cst in the $\ell_1$-metric in VLSI design. Therefore, there has been a lot of interest in understanding the complexity of \cst in $\ell_p$-metrics.  

Garey, Graham, and Johnson showed that Euclidean \cst is \np-hard, even in the plane \cite{Garey_Graham_Johnson_1977}.  Garey and Johnson also showed \np-hardness of \cst in the $\ell_1$-metric \cite{Garey_Johnson_1977} (also in the plane). Foulds and Graham extended this further, showing \np-hardness \cst in the Hamming metric \cite{foulds1982steiner}. These intractability results lead to the search for efficient approximation algorithms. 

In a remarkable collection of works, a \ptas was established for \cst in $\ell_1$ and $\ell_2$-metrics in constant dimensions (indeed, Arora's \ptas extends to all $\ell_p$-metrics) \cite{Arora_1998,Mitchell99,RaoS98}. However, it was unclear  if  \cst in $\ell_1$ and $\ell_2$-metrics in high dimensions admitted a \ptas. Trevisan \cite{Trevisan00} answered this question in the negative for the $\ell_1$-metric, by showing that \cst in $\ell_1$-metric is \apx-hard \cite{Chakrabarty_Devanur_Vazirani_2011,Hanan_1966,Day_Johnson_Sankoff_1986,Wareham_1995}. He left it as an open problem to show a similar hardness for the Euclidean \cst in high dimensions.

\begin{openq}\label{open2}
Does Euclidean \cst on $n$ terminals in $\Omega(\log n)$ dimensions admit a \ptas or is it \apx-hard?
\end{openq}

If a geometric optimization problem is \apx-hard, it is often easiest to establish that hardness in the $\ell_{\infty}$-metric. Intuitively, this arises due to the existence of the Fr\'{e}chet embedding, an isometric embedding of discrete metric spaces into the $\ell_{\infty}$-metric. However, it remained unknown whether high-dimensional \cst in the $\ell_p$-metric is \apx-hard  for \textit{any} $p > 1$, including $p = \infty$. Morally, we should not be able to understand hardness of approximation of Euclidean \cst until at least understanding the hardness of approximation of \cst in the $\ell_{\infty}$-metric.
 
 \begin{openq}\label{open3}
Does \cst in the $\ell_{\infty}$-metric on $n$ terminals in $\Omega(\log n)$ dimensions admit a \ptas or is it \apx-hard?
\end{openq}

We now shift our focus to \dst in $\ell_p$-metrics. Most approximation algorithms for \cst in the $\ell_p$-metric work by reducing them to an instance of \dst in the $\ell_p$-metric (including, for example, Arora's \ptas \cite{Arora_1998}). Therefore, from a hardness point-of-view, studying \dst in $\ell_p$-metrics can also be seen as a stepping stone to understanding \cst in $\ell_p$-metrics. Bartal and Gottlieb \cite{Bartal_Gottlieb_2021} showed that a \ptas exists for \dst in spaces of bounded doubling dimension (including, for example, \dst in $\ell_p$-metrics in constant dimensions) and Euclidean space up to $O(\sqrt{\log \log n})$ dimensions, where $n$ is the number of input points. Nonetheless, the inapproximability in higher dimensions remains open. For example, we can ask the following discrete variant of Open~Question~\ref{open2}.

\begin{openq}\label{open4}
Does Euclidean \dst on $n$ terminals in $\Omega(\log n)$ dimensions admit a \ptas or is it \apx-hard?\end{openq}

String metrics are another important family of metrics for \dst, particularly given applications of the Steiner tree problem to the study of phylogenetic trees in computational biology and to computational linguistics. Several works have studied the computational hardness of Steiner tree in the Hamming metric, establishing \apx-hardness in high-dimensions and \np-hardness for low dimensions with large alphabets \cite{Wareham_1995, althaus2013low, sankoff1975minimal, foulds1982steiner, fernandez1998approximability}.  Other popular string metrics less well understood in the context of the Steiner tree problem include the edit distance metric \cite{lev1966binary} and the Ulam metric.
\begin{openq}\label{open5}
Does \dst in the edit distance metric admit a \ptas or is it \apx-hard? What about the Ulam metric?
\end{openq}

We remark that exploring the discrete avenue to better understand the continuous version has been pursued for another celebrated geometric problem: \emph{clustering}. The discrete and continuous versions of minimizing popular clustering objectives such as $k$-means and $k$-median has been studied by the algorithmic community extensively \cite{byrka2017improved,DBLP:journals/siamcomp/AhmadianNSW20,CNS22,CK19,CKL22}.

\subsection{Our results}\label{sec:results}

Our first contribution is the resolution of Open Question~\ref{open3}. 

\begin{theorem} \label{thm:introellinf}
There is an efficiently computable function mapping a graph $G$ of order $n$ to an instance of \cst in the $\ell_{\infty}$-metric such that the optimal cost of the Steiner tree is $(n + \chi(G))/2$, where $\chi(G)$ is the chromatic number of $G$. Consequently, \cst in the $\ell_{\infty}$-metric is \apx-hard.
\end{theorem}

Our result also establishes a neat connection between the chromatic number of a graph and the cost of the optimal Steiner tree of the mapping of the graph. We remark that our current hardness results are derived from the hardness of graph coloring on graphs with linear chromatic number. Although the current hardness is only enough to imply \apx-hardness of \cst in the $\ell_{\infty}$-metric, stronger inapproximability results for graph coloring immediately translate to more robust results for \cst. For example, if the best known approximation algorithm for clique cover on cubic graphs were shown to be optimal, this would yield hardness of approximation within a factor of $17/16$ (they show a $5/4$-approximation algorithm in \cite{Cerioli_Faria_Ferreira_Martinhon_Protti_Reed_2008}).

Our second contribution is the resolution of Open~Question~\ref{open4}. In fact, we prove the \apx-hardness of \dst for all $\ell_p$-metrics. 

\begin{theorem}[Implied by Theorem \ref{thm: hardness of lp dst}]
   \label{thm:introapx} Let $p\in \mathbb{R}_{\ge 1}\cup\{\infty\}$. There is some constant $\varepsilon>0$ such that \dst on $n$ terminals is \np-hard to approximate to within a $(1+\varepsilon)$ factor in the $\ell_p$-metric, even in $O(\log n)$ dimensions. 
\end{theorem}

In fact, it is sufficient to prove that \dst is \apx-hard in the Euclidean metric, and then simply use the near isometric embedding of $\ell_2$-metric to other $\ell_p$-metrics \cite{Racke_2006} to obtain that  \dst is \apx-hard in every  $\ell_p$-metric.

That said, our result is stronger than suggested in Theorem~\ref{thm:introapx}, in the following sense. Our reduction is from a hybrid variant of the set cover and set packing problems (both restricted to the case when the input is a  collection of sets each of size exactly 3). Although current known hardness of bounded set cover and set packing yield very small constant factor inapproximability bounds for \dst in the $\ell_p$-metric, these problems are expected to be harder than currently known. We explicitly compute inapproximability factors throughout so that improvements for those problems will also immediately translate to improved hardness of approximation factors for \dst in $\ell_p$-metrics.

Elaborating, our reduction is from the $(\varepsilon,\delta)$-\setpT problem, where given as input a collection of sets $\mathcal{S}$ over universe $[n]$, we would like to distinguish between the completeness case where there exists $n/3$ pairwise disjoint sets in $\mathcal{S}$ that cover the whole universe $[n]$, and the soundness case where every set cover of $[n]$ from $\mathcal{S}$ must be of size at least $(1+\delta)(n/3)$ and every subcollection of pairwise disjoint sets in $\mathcal{S}$ can cover at most a $(1-\varepsilon)$ fraction of the universe. Note that if $(\varepsilon,\delta)$-\setpT problem is \np-hard then we may assume $\delta\in [\varepsilon/2,2\varepsilon]$ (see Remark~\ref{rem:epsdelta} for an explanation). 

We derive the following hardness of approximation factors, in terms of the parameters of the reduction from \setpT to \dst, from our proof of Theorem~\ref{thm: hardness of lp dst}.

\begin{theorem}\label{thm:intromain}
    Let $p\in \mathbb{R}_{\geq 1}\cup\{\infty\}$. Assuming $(\varepsilon,\delta)$-\setpT is \np-hard, we have that \dst in $\ell_p$-metric is \np-hard to approximate to $(1+\gamma)$ factor, where
    $$\gamma:=\begin{cases}
        \delta/4&\text{ if } p=\infty \\
        \frac{\varepsilon}{2}\left(1-\frac{1}{3^{1/p}}\right)+2\delta\left(\frac{1}{2\cdot 3^{1/p}}-\frac{3}{8}\right)&\text{ if }p>1/\log_3 (4/3) \\
        \varepsilon/8&\text{ if } p=1/\log_3 (4/3) \approx 3.8\\ 
 \varepsilon/26 &\text{ if } p  \in [1,  1/\log_3(4/3))
     \end{cases}$$
\end{theorem}

We have only tried to optimize our reduction for $p=2$ and $p>1/\log_3(4/3)\approx 3.8$, and, even then, we only optimize fully in the case of $\delta = \varepsilon/2$. At each of these thresholds, the scaling of our \dst hard instances changes. We optimize in the range $p>1/\log_3(4/3)$ because the optimal scaling of our \dst hard instance is clear. Indeed, $p = 1/\log_3(4/3)$ is exactly the threshold where the interactions between the parameters become far more complex. It is possible to obtain a better relationship between $\varepsilon$, $\delta$, and $\gamma$ for $p\in (1,1/\log_3(4/3))\setminus \{2\}$ by carefully analyzing the constraints mentioned in Theorem~\ref{thm: hardness of lp dst}. 

In Figure~\ref{fig: eps delta depend}, we capture the interplay between $\varepsilon,\delta,$ and $\gamma$. The purple region indicates that, between the set packing and set cover (actually vertex cover) problems, the reduction (i.e., embedding) of the set system only goes through by starting from the covering problem and not from the packing problem. The green region on the other hand is the exact opposite, and the reduction  of the set system to the \dst only goes through by starting from the packing problem and not from the covering problem. In the yellow region, the behavior is unclear, in part because we did not carefully analyze the dependency between $\varepsilon$ and $\delta$. Finally, in the red region, we benefit from both the covering and packing hardness. This is elaborated further in Section~\ref{sec:figureexp}. 

\begin{figure}[!h]
    \centering

	\includegraphics[width=\textwidth]{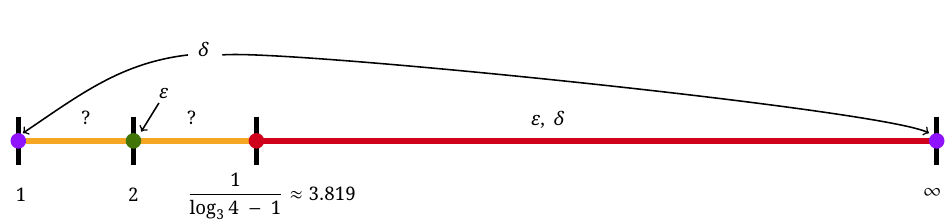}

    \caption{Color-coded range of $\ell_p$ hardness of approximation dependency on $\varepsilon$ and $\delta$.}
    \label{fig: eps delta depend}
\end{figure}

Moreover, we note that we can derive explicit constants for the hardness of approximation factors, and we do for \dst in the Hamming ($\ell_0$), $\ell_1,\ell_2$-metrics in Appendix~\ref{sec:explicit}.

Our third contribution is observing  the relative computational hardness of \dst with respect to  \cst. Although solving \dst instances is often used as a subroutine in algorithms for \cst, prior to this work it was unclear whether solving \dst problems was computationally easier or harder in general. We show that \dst is essentially at least as hard as \cst.

\begin{theorem}[Implied by Theorem \ref{thm: reduction from cst to dst}]\label{thm:csttodstintro}
Let $p \in [1,\infty]$ and $\alpha>0$. If \cst in the $\ell_p$-metric is \np-hard to approximate within a factor of $(1 + \alpha)$, then, for all $\varepsilon > 0$, \dst in the $\ell_p$-metric is \np-hard to approximate within a factor of $(1 + \alpha - \varepsilon)$.
\end{theorem}

We prove this by showing that we can efficiently compute  a collection of candidate Steiner points for a \dst instance such that the optimal tree for the \dst instance is near-optimal for the corresponding \cst instance. Hence, efficient approximation algorithms for \dst yield similarly effective approximation algorithms for \cst. As an immediate consequence, this yields a $1.39$-approximation polynomial time algorithm for high dimensional \cst in $\ell_p$-metrics (see Corollary~\ref{cor:alg}). The proof heavily relies on structural results about near-optimal Steiner trees from \cite{DuZhangFeng91, Borchers97, Bartal_Gottlieb_2021}. 

 \paragraph{String Metrics.} By applying the isometric embedding of the 
Hamming metric to the Ulam metric (for example see Lemma 4.5 in \cite{DBLP:journals/corr/abs-2112-03222}) we can prove \apx-hardness of \dst in the Ulam metric. 
\begin{theorem} \label{thm: ulam hardness}
There is some constant $\varepsilon > 0$ such that \dst in Ulam metric on $n$ terminals is \np-hard to approximation to within a $(1+\varepsilon)$ factor, even on permutations of length $O(\log n)$. 
\end{theorem}
Similarly, by applying the near-isometric embedding of the Hamming metric to the edit distance metric (for example, see Section 4.1 in \cite{DBLP:conf/stoc/Rubinstein18}), we can prove \apx-hardness for \dst in the edit distance metric  as well. 
\begin{theorem} \label{thm: edit hardness}
There is some constant $\varepsilon > 0$ such that \dst in edit distance metric on $n$ terminals is \np-hard to approximation to within a $(1+\varepsilon)$ factor, even on Boolean strings of length $O(\log n \log \log n)$. 
\end{theorem}
In combination, these results resolve Open Question \ref{open5}. See Appendix~\ref{sec: string metrics} for details.

 \paragraph{Dimensionality Reduction.}
Our proof of Theorem~\ref{thm:introapx} holds when the dimension of the point set has linear dependency in the size of the point set. However, we can use the dimensionality reduction technique introduced in \cite{CK19} and generalized in \cite{CKL22} to obtain the same inapproximability result as in Theorem~\ref{thm:introapx}, even when the dimension of the point set has logarithmic dependency on the size of the point set. We discuss these dimensionality reduction details (for constant $p$) in Appendix~\ref{sec: dimensionality reduction}.\footnote{Alternatively, it is possible to apply the Johnson-Lindenstrauss lemma \cite{Johnson_Lindenstrauss_1984} in the $\ell_2$-metric and then use a near-isometric embedding into the $\ell_p$-metric for $p \in [1,\infty]$ to achieve similar guarantees. We thank the anonymous referee for this suggestion.}
Assuming the \emph{Exponential Time Hypothesis} \cite{IP01,IPZ01}, the dimension cannot be reduced much below logarithmic in size of point set as the runtime of the approximation scheme of \cite{Arora_1998,Mitchell99,RaoS98} would then be subexponential.

 \subsection{Our Techniques}\label{sec:techniques}

In this subsection, we contextualize our techniques and provide a high level description of our technical contributions. We emphasize that the conceptual difficulty in the proofs of Theorems~\ref{thm:introellinf},~\ref{thm:introapx},~and~\ref{thm:intromain} is devising the appropriate starting problem and a clean embedding. The main technical challenge is to prove the soundness property of the reduction. This involves proving several structural properties of optimal Steiner trees in the corresponding metrics.

\subsubsection{\texorpdfstring{Proof Overview of Theorem~\ref{thm:introellinf}: Hardness of Approximation of \cst in $\ell_{\infty}$-metric}{Proof Overview of Theorem~\ref{thm:introellinf}: Hardness of Approximation of CST in linf-metric}}

First, we discuss the context of Open Problems~\ref{open2}~and~\ref{open3} and describe our techniques for proving Theorem~\ref{thm:introellinf} (the resolution of Open Problem~\ref{open3}). 

\paragraph{\texorpdfstring{Prior Work on \cst and Technical Difficulties.}{Prior Work on CST and Technical Difficulties.}}
Prior to this work, the only \apx-hardness results for high dimensional \cst were in the Hamming $(\ell_0)$ and Rectilinear $(\ell_1)$ metrics \cite{Wareham_1995, Trevisan00}. The discrete structure of Hamming \cst makes it far more amenable to existing tools, such as reductions from hard graph problems, than other variants of \cst. Indeed, Trevisan's proof of \apx-hardness of Rectilinear \cst is a direct reduction from Hamming \cst. While he appeals to the integrality property of the Min Cut LP relaxation, Hanan showed decades earlier that instances of Rectilinear \cst admit optimal Steiner trees using only Steiner points on the grid formed by the intersection of axis parallel lines passing through the terminals \cite{Hanan_1966}. This \textit{Hanan grid} arises from the ``coordinate independence'' in the $\ell_1$-metric, i.e., the fact that a change in one coordinate changes the $\ell_1$ distance by the same amount. This property is unique to the $\ell_1$-metric among $\ell_p$-metrics.

In other $\ell_p$-metrics, we know much less. From \cite{Garey_Graham_Johnson_1977}, we know that Euclidean \cst is \np-hard. The proof is shown in the plane, making heavy use of the fact that optimal Euclidean Steiner trees have non-intersecting edges and those edges meet at at least 120 degree angles \cite{Gilbert_Pollak_1968}. However, the non-intersecting edge property is not especially useful for point configurations outside of the plane. Without this property, we have very few strong tools for restricting the structure of optimal Steiner trees. Indeed, we do not know a proof of \np-hardness of Euclidean \cst which does not involve directly appealing to hardness in the plane. This poses a major barrier to showing \apx-hardness of Euclidean \cst. In low dimensions, we know a \ptas exists \cite{Arora_1998}, so we cannot appeal to \np-hardness techniques. Hence, showing \apx-hardness in high-dimensions will ultimately require a brand new \np-hardness argument.

Euclidean \cst has also been heavily studied from the perspective of discrete geometry. The central quantity of interest in that setting is the \textit{Steiner ratio} of a point configuration: the ratio of the cost of an optimal Steiner tree and a minimum spanning tree. For example, the Steiner ratio of the three vertices of a unit equilateral triangle is $\sqrt{3}/2$ with the optimal Steiner tree achieved by adding the center of the triangle as a Steiner point and connecting it to the vertices. Perhaps the most basic question about the Steiner ratio is the following: what is the minimum Steiner ratio over point configurations in the Euclidean plane? Gilbert and Pollak conjectured that it is exactly $\sqrt{3}/2 \approx 0.866$, coming from the equilateral triangle \cite{Gilbert_Pollak_1968}. This conjecture, the Gilbert-Pollak Steiner ratio conjecture, remains open after more than half a century \cite{ivanov2012steiner}. The best known lower bound on this quantity is $\approx 0.824$, leaving a major gap, even in the plane \cite{Chung_Graham_1985}. 

Predictably, in higher dimensions we know even less about the Steiner ratio and the construction of Steiner trees. Although extrapolating from the plane one might guess that the vertices of a regular simplex might admit the minimum Steiner ratio in higher dimensions (the Generalized Gilbert-Pollak Steiner ratio conjecture), Du and Smith disproved this conjecture without providing a satisfying new candidate point configuration \cite{Du_Smith_1996}. Moreover, we do not know how to construct optimal Steiner trees of high-dimensional point configurations, not even of the vertices of regular simplices (although \cite{Chung_Gilbert_1976} gives a candidate optimal Steiner tree for some dimensions). As in the plane, we also have a considerable gap in the bounds on the Steiner ratio: the minimum lies somewhere between $0.669$ and $0.62$ \cite{Chung_Gilbert_1976, Du_Smith_1996, Du_Wu_Lu_Xu_2011}. 

Given our weak bounds, we cannot appeal to the Steiner ratio in analyzing delicate inapproximability arguments. The fact that we cannot even construct high-dimensional optimal Steiner trees compounds this difficulty further. The state of affairs is similar in other $\ell_p$-metrics. New results about the inapproximability of high-dimensional \cst require novel structural insights.

\paragraph{\texorpdfstring{Reduction from Graph Coloring to \cst in the $\ell_{\infty}$-metric.}{Reduction from Graph Coloring to CST in the l_inf-metric.}}

In this work, we show \apx-hardness for \cst in the $\ell_{\infty}$-metric. Our proof is unlike the results in \cite{Wareham_1995, Trevisan00}. Unable to reduce the problem from a discrete metric space (e.g., Hamming space), we rely instead on the observation that the distance between any two points in the $\ell_{\infty}$-metric is determined by a single coordinate. We use this observation to introduce structural guarantees on Steiner trees we construct in the $\ell_{\infty}$-metric.

Given a graph, we arbitrarily orient its edges to form the digraph $G(V, A)$. We construct an instance of \cst in the $\ell_{\infty}$-metric in $\R^{|A|}$ as follows. Each coordinate of the terminals corresponds to an edge in $A$. The terminals themselves correspond to vertices, along with the usual addition of a root terminal $\mathbf{0}$, the all-zeroes vector. Given $v \in V$ and $a \in A$, the terminal $t$ corresponding to $v$ is $1$ in the coordinate corresponding to $a$ if $a$ is an outgoing edge from $v$, $-1$ if $a$ is an incoming edge to $v$, and $0$ otherwise. This ensures that for each coordinate, there are exactly two terminals that are nonzero in that coordinate. The terminals corresponding to adjacent vertices are at distance $2$, and the terminals corresponding to non-neighboring vertices are at distance $1$.

Then, a minimum proper vertex coloring $\pi: G \to [\chi(G)]$ induces a low-cost Steiner tree of these terminals. For each color class of vertices, there is a Steiner point $s$ at distance $1/2$ from the set of corresponding terminals of that color class. Among those terminals, for each coordinate $i$, there is at most one terminal nonzero in that coordinate (using that each color class is an independent set). Setting $s_i$ to be $0$ or half of that nonzero coordinate yields the desired Steiner point (thus, we have $\|s\|_\infty=1/2$). The tree  is then constructed by connecting the terminals corresponding to each color class to the  Steiner point described above and then connecting each of those Steiner points to $\mathbf{0}$. Each edge in the tree is of cost exactly $1/2$ and there are $n + \chi(G) + 1$ total vertices, yielding a total cost of $(n + \chi(G))/2$. 

It turns out that this simple tree is actually the optimal Steiner tree for this configuration! We prove this in Theorem \ref{thm:introellinf}. Upon showing this fact, the hardness of graph coloring (even in graphs with linear chromatic number) yields \apx-hardness of \cst in the $\ell_{\infty}$-metric. 

To do this, we show that there exists an optimal Steiner tree exhibiting many complementary properties and then prove that this tree is precisely the tree described above. At a high level, the objective is to prove two main claims: nonzero terminals are leaf nodes in the Steiner tree, and no Steiner point is adjacent to another Steiner point. Upon showing these two claims, it is easy to check that the Steiner points in the optimal tree are of the desired form.

Underlying the arguments required for each of these claims is an important observation about the structure of optimal Steiner trees of this \cst instance. Let $T$ be an optimal Steiner tree of this \cst instance, and let $t, t'$ be terminals such that $t_i = 1$ and $t_i' = -1$.  Observe that $t$ and $t'$ are the only terminals in $T$ with nonzero $i\textsuperscript{th} $ coordinates. We can show that there exists an optimal Steiner tree such that the $i\textsuperscript{th} $ coordinate of Steiner points  on the path between $t$ and $t'$  decreases maximally away from $t$ until the coordinate reaches $0$ (in the sense of determining the distance between adjacent points). Likewise, the $i\textsuperscript{th} $ coordinate increases maximally toward $0$ along edges away from $t'$. Throughout, the coordinates reduce in magnitude maximally away from the path linking $t$ and $t'$. If, for example, the path between $t$ and some Steiner point $s$ passes through $\mathbf{0}$, then $s$ cannot have any coordinates of the same sign as $t$. Adopting this view of the points in $T$ being projected onto the $i\textsuperscript{th} $ coordinate is valuable intuition for the proofs.

For example, we use this intuition in proving that we may assume nonzero terminals are leaf nodes. It is easy to show that in some optimal Steiner tree, the nonzero terminals are non-adjacent, but, a priori, it is unclear whether they may be adjacent to many Steiner points. Suppose that $t$ is a nonzero terminal adjacent to multiple Steiner points. Consider removing the edges from $t$ to Steiner points such that their resultant connected components do not contain $\mathbf{0}$. The idea is that, in these new connected components, coordinates of Steiner points with the same sign as coordinates of $t$ are no longer useful, so we can set them all to zero. Upon doing so, connecting the Steiner points originally adjacent to $t$ to $\mathbf{0}$ reconnects the tree without increasing the cost. This amounts to ensuring that $\mathbf{0}$ is included in the path between these Steiner points and $t$.

Showing that the Steiner points in the tree are non-adjacent is the most technical part of the proof. To do so, we first consider \textit{Steiner leaves}, the Steiner points in the Steiner tree that become leaf nodes upon removing all nonzero terminals in the tree. One important property of these Steiner points is that they are necessarily adjacent to more terminals than Steiner points (it is easy to show that Steiner points must be degree at least $3$ by the triangle inequality). It is straightforward to show that terminals adjacent to the same Steiner point cannot overlap in nonzero coordinates. Then, to minimize the total cost of the edges incident to the Steiner leaf, its coordinates must have magnitude at most $1/2$, with magnitude $1/2$ in each of the nonzero coordinates of the adjacent terminals (with the same sign as that nonzero coordinate). Using this fact we can show that the terminals adjacent to adjacent Steiner leaves cannot overlap on nonzero coordinates either (by removing the edge between the Steiner leaves and adding a new edge to $\mathbf{0}$). Indeed, an even stronger fact is true: if two Steiner leaves are adjacent to a common Steiner point, then the terminals adjacent to the Steiner leaves cannot overlap on nonzero coordinates. This follows from the usual trick of dropping an edge and connecting one of the Steiner leaves directly to $\mathbf{0}$.

By leveraging these facts and the coordinate-wise view of the Steiner tree, we can show that each non-leaf Steiner point in the tree is adjacent to at least one other non-leaf Steiner point. To do this, consider $s$ and $s'$, two Steiner leaves adjacent to a common Steiner point $s''$. We can remove all edges from $s'$ to its adjacent terminals and connect them instead to $s$. We can then modify $s$ to ensure that all neighboring terminals are at distance $1/2$ away and then remove $s'$ entirely since it is a Steiner leaf. Repeating this process will result in one of two things. Either $s''$ will become degree $2$ and can be removed via the triangle inequality or $s''$ is adjacent to another non-leaf Steiner point. After repeating for all Steiner points adjacent to Steiner leaves, all remaining non-leaf Steiner points will be adjacent to at least one other non-leaf Steiner point.

Finally, consider removing all nonzero terminals and Steiner leaves. The resulting tree has at most vertex of degree $1$, $\mathbf{0}$, and, hence, is only a tree if all Steiner points were Steiner leaves. That is, all Steiner points in the tree must have only been adjacent to terminals. 

\subsubsection{\texorpdfstring{Proof Overview of Theorems~\ref{thm:introapx}~and~\ref{thm:intromain}: Hardness of Approximation of \dst in $\ell_p$-metrics}{Proof Overview of Theorems~\ref{thm:introapx}~and~\ref{thm:intromain}: Hardness of Approximation of DST in lp-metrics}}

We present below the techniques needed to prove Theorems~\ref{thm:introapx}~and~\ref{thm:intromain}, first establishing the appropriate technical context.
\paragraph{Known Reductions from Set Cover and Vertex Cover.} 
The following is a gap-preserving reduction from Set Cover to \dst in general metrics due to Karp \cite{karp1972reducibility}. 
Given a set system $([n],\mathcal{S})$, where $\mathcal{S}$ is a collection of $m$ subsets of $[n]$, we construct a \dst instance as follows. Our terminals will be $[n]\cup\{r\}$ (a special root vertex is introduced). Our facilities will be the subsets in $\mathcal{S}$. We connect each subset to $r$ and also have edges between subsets/facilities and terminals/universe elements based on their membership. See Figure \ref{fig: karp reduction}. The metric is simply the shortest path metric on the graph. In the completeness case\footnote{Here we use  terminology from the hardness of approximation literature to refer to the two cases of our decision problems. The completeness case is the case of our transformed problem corresponding to the ``yes'' case of our original instance. The soundness case corresponds to the ``no'' case.} there is a set cover of size $k$, and that corresponds to a Steiner tree of cost $n+k$. In the soundness case, if every set cover of $[n]$ must be of size ($1+\delta)k$, (for some $\delta>0$), then one can show that the Steiner tree cost is at least $n + k(1+\delta)$. By then considering set systems where each set contains exactly 3 elements and in which the set cover is a partition in the completeness case, we obtain $k=n/3$, and thus \dst is hard to approximate to within a  $(1+\delta/4)$ factor.\footnote{Although better hardness of approximation results are known for set cover for sets of size $B \geq 4$, reductions based on these set systems translate to hardness of approximation factors of roughly $\frac{B+ 1 + \delta}{B + 1}$, where we have hardness of approximation within a factor of $(1 + \delta)$. The value of $\delta$ here grows much more slowly than $B$ \cite{trevisan2001non}, so it is generally better to consider sets of size $3$, the minimum size where we can achieve hardness.}
\begin{figure}[h]
    \centering

	\includegraphics[width=0.6\textwidth]{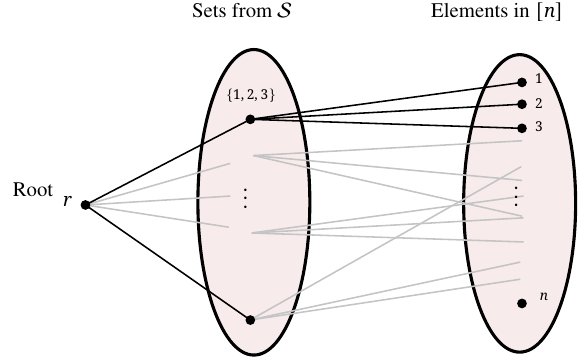}

    \caption{Instance of \dst in general metrics from Karp's reduction from Set Cover}
    \label{fig: karp reduction}
\end{figure}

To obtain \apx-hardness for \dst in the Hamming (or $\ell_1$) metric, we can simply start from a special set system, namely the vertex cover problem on cubic graphs, and then embed the graphs into the Hamming space as follows. Given a vertex cover instance, $G(V,E)$, our facilities are $\mathbf{e}_1,\ldots, \mathbf{e}_{|V|}$, and for every edge $(u,v)$ we have the terminal $\mathbf{e}_u+\mathbf{e}_v$ (where $\mathbf{e}_i$ is a standard basis vector). We have a special additional terminal (much like the reduction from Set Cover in the previous paragraph) which is the all-zeroes vector. The completeness and soundness case go through like in the set cover instance, although some care needs to be taken in the completeness case. In particular, we cannot assume that the minimum vertex cover forms an independent set in the completeness case as the problem of deciding whether that is the case is in $\textsf{P}$.

We can then try to similarly reduce from vertex cover to \dst in the Euclidean metric, but the reduction fails in a suggestive manner. Consider using the same choice of facilities and terminals as in the Hamming case. Given a vertex cover of size $k$, there is a Steiner tree of cost $m + k$ where we choose the facilities corresponding to the vertices in the vertex cover. Each non-root terminal is a leaf node connected either to a facility at distance $1$ or directly to the all-zeroes vector. Each facility is then connected to the all-zeroes vector. 

However, unlike in the Hamming metric, such a Steiner tree might be far from optimal. Crucially, if a facility is adjacent to two non-root terminals, it is also more efficient to connect each of those adjacent terminals directly to the all-zeroes vector, incurring a saving of $3 - 2\sqrt{2}$ in the Steiner tree cost. 

It is relatively straightforward to show that there are minimal cost Steiner trees in which all non-root terminals are leaf nodes and either adjacent to a facility of maximum degree or the all-zeroes vector. This is notably the same as the construction outlined above upon removing facilities adjacent to two non-root terminals. However, the minimum cost of a Steiner tree on these facilities and terminals no longer only depends on the minimum size of a vertex cover in $G$. Instead, the cost of the Steiner tree is minimized by finding a maximum independent set, then taking the corresponding Steiner points, and connecting all terminals corresponding to edges not incident to the independent set directly to the all-zeroes vector. Our reduction from minimum vertex cover naturally transformed into a reduction from maximum independent set!

Importantly, when we consider this reduction for even larger $p$, larger than around $p \approx 2.409$, the Steiner points lose all of their utility and even facilities adjacent to three non-root terminals do not help reduce the cost. By using a more robust reduction from set packing, the analog of maximum independent set on graphs, we manage to avoid this issue.

\paragraph{\texorpdfstring{Reduction from Set Packing to Euclidean \dst.}{Reduction from Set Packing to Euclidean DST.}}
Our main insight is that while we cannot reduce from the vertex cover problem to Euclidean \dst, we can however reduce from the independent set problem on cubic graphs to Euclidean \dst.


As such, we work with set packing problem, a generalization of the independent set problem. In particular, we work with \setpT, the set packing problem where each set has cardinality $3$, and it becomes much easier to limit the adjacencies of Steiner points when they depend on fewer coordinates.  We also assume that there is a set partition in the completeness case. Given a set system $([n], \mathcal{S})$ where $\mathcal{S}$ is a collection of subsets of $[n]$ of cardinality $3$, we construct an instance of Euclidean \dst as follows. The terminals are $\{ \mathbf{e_i} \,:\, i \in [n]\} \cup \{\mathbf{0}\}$ and correspond exactly to the elements of the universe $[n]$, besides the addition of a special root terminal $\mathbf{0}$. The facilities are $\{\frac{1}{6}(\mathbf{e_i} + \mathbf{e_j} + \mathbf{e_k}) \,:\, \{i,j,k\} \in \mathbf{S}\}$. Each facility corresponds exactly to a set in $\mathcal{S}$; it is a scaling of the characteristic vector of the set.

In the completeness case, there is a set packing. Choosing the facilities corresponding to the sets in the packing yields a Steiner tree of cost $(\sqrt{3}/2 + 1/\sqrt{108})n$. Each universe element is connected to the facility such that the corresponding universe element is contained in the corresponding set in the packing. All facilities are connected to the root terminal $\mathbf{0}$. In the soundness case, suppose that every set packing of $([n], \mathcal{S})$ can cover at most $(1-\varepsilon)n$ elements. One can show that there exists a minimum cost tree such that: (1) the choice of facilities correspond to a maximal set packing $\mathcal{S}' \subseteq \mathcal{S}$, (2) each facility is adjacent exactly to $\mathbf{0}$ and the three terminals corresponding to universe elements contained in the facility's set, and (3) each non-root terminal is a leaf node and is either adjacent to a Steiner point or the root terminal $\mathbf{0}$. Given this structure, the cost of the tree is at least $((1 - \varepsilon)(\sqrt{3}/2 + 1/\sqrt{108}) + \varepsilon)n$, yielding hardness of approximation within a factor of  $1 + \varepsilon/26$. The scale factor of $1/6$ in the contruction of the facilities is chosen carefully to ensure $(1)$, $(2)$, and $(3)$ and optimize this hardness of approximation factor (See Section~\ref{sec: l_p} for details.). 

This reduction is natural when the terminals are universe elements and the Steiner points are sets. By having our facilities depend on multiple coordinates, we keep the terminals far apart and also keep Steiner points relatively far apart (at least relative to our root) so that we can limit their adjacencies. Thus, upon removing the zeroes vector (special terminal), the Steiner tree splits into many connected components which can be interpreted as a choice of sets for our set packing.

\paragraph{Reduction from Set Packing to \texorpdfstring{\dst}{DST} in \texorpdfstring{$\ell_p$}{lp}-metric Spaces.}
It turns out that the reduction from \setpT to Euclidean \dst generalizes  to a reduction from \setpT to \dst in $\ell_p$-metric spaces. There are two main differences. 

First, we modify our hard instance of \setpT slightly to become $(\varepsilon, \delta)$-\setpT. The completeness case is identical to that outlined in the Euclidean \dst case above. In the soundness case, in addition to each set packing of $([n], \mathcal{S})$ covering at most $(1-\varepsilon)n$ elements, we also have that, for any set cover $\mathcal{S}' \subseteq \mathcal{S}$ such that $[n] \subseteq \cup_{S \in \mathcal{S}'} S$, $|\mathcal{S}'| \geq (1 + \delta)(n/3)$. We add this restriction on the size of minimum set covers since we observe that, for certain $p$, the hardness of approximation of \dst in $\ell_p$-metric spaces resultant from this reduction depends on both $\varepsilon$ and $\delta$ (or even only $\delta$). See Figure \ref{fig: eps delta depend} and the surrounding discussion.

Additionally, in the reduction to \dst in the $\ell_p$-metric, the scale factor of $1/6$ in the Euclidean setting becomes some constant depending on $p$. In part, this is to sufficiently restrict the structure of optimal Steiner trees in the soundness case. Different choices of scale factor also yield improved hardness of approximation ratios. For example, it is optimal to use a scale factor close to $1/2$ for very large $p$.

\subsubsection{Proof Overview of Theorem~\ref{thm:csttodstintro}:  Reduction from \texorpdfstring{\cst}{CST} to \texorpdfstring{\dst}{DST}}

Finally, we discuss the techniques needed to prove Theorem~\ref{thm:csttodstintro}. 
\cst and \dst have each independently garnered considerable interest, but, prior to this work, the relative computational hardness of the problems was unclear.\footnote{In fact, to the best of our knowledge, the distinction between the discrete and continuous variants was never formalized and systematically studied in the literature.} On the one hand, searching over a possibly uncountable collection of candidate Steiner points in \cst appears daunting. On the other hand, perhaps it is possible to select collections of candidate Steiner points for \dst that more effectively emulate other computationally hard problems. However, as the other side of the same coin, the continuous version allows the algorithm designer the freedom to pick Steiner points anywhere in space, whereas the discrete variant restricts the use of the  geometry of the space by only allowing Steiner points from the set of facilities.\footnote{These remarks also hold (in spirit) for other geometric optimization problems. For a useful comparison, consider clustering. Rather surprisingly, in \cite{CKL21}, the authors showed that, in the $\ell_\infty$-metric space, the continuous version of clustering is possibly harder than the discrete version. }

In this paper we show that, up to an arbitrarily small factor, \dst is at least as hard to approximate as \cst (see Theorem~\ref{thm:csttodstintro} above). The proof uses an important structural result of optimal Steiner trees proved in a recent work of Bartal and Gottlieb \cite{Bartal_Gottlieb_2021}. (Similar structural results were proved in the context of $k$-restricted Steiner trees by Du, Zhang, and Feng, and Du and Borchers \cite{DuZhangFeng91, Borchers97} in the past and these results would suffice for our purpose as well). Namely, there exist near-optimal Steiner trees composed of optimal Steiner trees of constant size, linked by edges between terminals. By embedding all small subsets of terminals in low dimensional spaces and computing optimal Steiner trees using existing \ptas's, we can then construct a set of candidate Steiner points for \dst (by unioning all of the Steiner points in the optimal Steiner trees of these small subsets). Then, the optimal Steiner tree on the \dst instance approximates the cost of the Steiner tree on the \cst instance arbitrarily well. As a consequence, this yields the best known approximation algorithm for high-dimensional \cst (Corollary~\ref{cor:alg}). 

\subsection{Organization of the Paper}
In Section~\ref{sec:prelim} we define the problems of interest to this paper. 
In Section~\ref{sec: l_inf cst hardness} we prove \apx-hardness of \cst in the $\ell_{\infty}$-metric via a reduction from graph coloring. In Section~\ref{sec:metricspace}, we provide our generic framework of moving from a set system to an abstract and yet parameterized metric space to obtain hardness of \dst. 
In Section~\ref{sec: l_p}, we detail our $\ell_p$-metric embedding of the abstract metric space described in the previous section. 
In Section~\ref{sec: reduction from cst to dst}, we relate \cst to \dst, describing the reduction between the two problems. In Appendix~\ref{sec:explicit}, we prove explicit factors of inapproximability for \dst in $\ell_0,\ell_1,$ and $\ell_2$-metrics. In Appendix~\ref{sec: string metrics}, we explain how to derive inapproximability of \dst in string metrics using known results from the literature. In Appendix~\ref{sec: dimensionality reduction}, we show that our hardness results for \dst hold even in $O(\log n)$-dimensions using a specialized near-isometric embedding.

\section{Preliminaries}\label{sec:prelim}

In this section, we detail some notations and definitions related to $\ell_p$-metric spaces, the Steiner tree problem, the graph coloring problem, the set packing problem, and the vertex cover problem.

\begin{paragraph}{$\ell_p$-metrics.} Let $p\in \reals$. For any two points $a,b \in \R^d$ we denote the distance between them in the $\ell_p$-metric by
\[
\| a - b \|_p := \begin{cases}
\left|\{i\in[d]| a_i\neq b_i\}\right|&\text{ if }p=0,\vspace{0.15cm}\\
\left( \underset{i\in[d]}{\sum} |a_i - b_i|^p \right)^{1/p}&\text{ if }p\in\mathbb{R}_{\ge 1},\vspace{0.15cm}\\
\underset{i \in[d]}{\max}\ |a_i-b_i|&\text{ if }p=\infty.
\end{cases}
\]
\end{paragraph}

\begin{paragraph}{Steiner tree problem.}  Let  $(\mathcal{X}, \Delta)$ be a metric space on set $\mathcal{X}$ and distance function $\Delta$. 
Given a set $P$ of points in $(\mathcal{X}, \Delta)$, a \textit{Steiner tree} of those points is a minimal spanning tree of $P \cup S$ for some $S\subseteq \mathcal{X}$. The initial points in $P$ are called \textit{terminals} and the points in $S$ are called \textit{Steiner points}. 

Given a tree $T$ with vertices as points in $\mathcal{X}$ and edge weights induced by the distance function $\Delta$, the cost of the tree $T = (P \cup S, E)$, denoted by $\cost_{\Delta}(T)$ is 
\[
 \cost_{\Delta}(T):=\sum_{e = (u,v) \in E} \Delta(u, v).
\]

In the continuous Steiner tree problem (\cst) in the metric space $(\mathcal{X}, \Delta)$, we are given as input $P \subseteq \mathcal{X}$, and the goal is to find a minimum cost Steiner tree of $P$. In the discrete Steiner tree problem (\dst) in the metric space $(\mathcal{X}, \Delta)$, we are given a pair $(P, X)$ as input, with $P \subseteq \mathcal{X}$ the set of $n$ terminals and $X \subseteq \mathcal{X}$ with $|X| = O(\poly(n))$ as the set of facilities. The goal is then to find a minimum cost of a Steiner tree $T = (P \cup S, E)$ of $P$ with  $S \subseteq X$.

\dst can also be considered as a graph problem. In this setting it is called the metric Steiner Tree problem on graphs. Given a complete graph $G = (P \cup X, E)$ with weight function $w : E \to \mathbb{R}^{+}$ satisfying the triangle inequality, the metric Steiner tree problem is to find the minimum cost of a Steiner tree of $P$ using Steiner points from the set $X$.
\end{paragraph}

\paragraph{Graph coloring problem.} In the $k$-coloring problem, we are given a graph $G(V,E)$ and a positive integer $k$, and the goal is to determine whether there is a \textit{proper coloring} $\pi: V \to [k]$ such that, for each $i \in [k]$, $\pi^{-1}(i)$ forms an independent set in $G$. That is, the vertices in $G$ can be partitioned into at most $k$ independent sets. The minimum $k$ such that $G$ admits a proper $k$-coloring is the \textit{chromatic number} of $G$ and is denoted $\chi(G)$.

In the $(a, b)$-coloring, the question is to instead decide which of the following two cases hold:
\begin{itemize}
    \item \textit{Completeness:} $G$ admits a proper $a$-coloring ($\chi(G) \leq a$).
    \item \textit{Soundness:} Every proper coloring of $G$ uses at least $b$ colors  ($\chi(G) \geq b$).
\end{itemize}

\begin{paragraph}{Set packing problem.}
We say that $(\mathcal{U}, \mathcal{S})$ is a set system if $\mathcal{S}$ is a collection of subsets of $\mathcal{U}$.
An instance of the Set Packing problem (\setp) is given by the triple $([n], \mathcal{S}, k)$, where $([n], \mathcal{S})$ is a set system with $|\mathcal{S}| = m$ and $k \in \mathbb{N}$, and the goal is to determine if there exists some $\mathcal{S}' \subseteq \mathcal{S}$ such that for all distinct $S, S' \in \mathcal{S}'$, we have $S \cap S' = \emptyset$ and   $|\cup_{S \in \mathcal{S}'} S| \geq k$? When all of the sets in $\mathcal{S}$ are fixed to be of size $B$, we abbreviate the problem as \setpB. 

We define $(\varepsilon, \delta)$-\setpT to be the related decision problem in which we have as input $([n], \mathcal{S})$ as above and decide which of the following two cases hold.
\begin{itemize}
    \item \textit{Completeness:} $[n]$ may be partitioned into $n/3$ sets all of which are in $\mathcal{S}$.
    \item \textit{Soundness:} Any collection of disjoint sets in $\mathcal{S}$ covers at most $(1-\varepsilon)n$ elements in $[n]$, and for any collection $\mathcal{S}' \subset \mathcal{S}$ (sets in $\mathcal{S}'$ need not be pairwise disjoint) such that $[n] \subseteq \cup_{S \in \mathcal{S}'} S$ satisfies $|\mathcal{S}'| \geq (1 + \delta)(n/3)$.
\end{itemize}

\begin{remark}In the soundness case, the parameters $\varepsilon$ and $\delta$ characterize the maximum coverage of a set packing and the minimum size of a set cover, respectively. These parameters are related. Indeed, given $\varepsilon$, we should have $\delta \in [\varepsilon/2, 2\varepsilon]$. If any collection of disjoint sets in $\mathcal{S}$ covers at most $(1-\varepsilon)n$ elements in $[n]$, then the smallest possible size of a set cover is $(1 + \varepsilon/2)(n/3)$ (if every additional set covers $2$ elements). Assuming a set packing covering $(1-\varepsilon)n$ elements is possible, the largest possible minimum set cover is $(1 + 2\varepsilon)(n/3)$ (covering the remaining elements one at a time).  \label{rem:epsdelta}\end{remark}
Our reduction will use the following \np-hardness of $(\varepsilon, \delta)$-\setpT.

\begin{theorem}[Theorem 4.4, \cite{Petrank_1994}] \label{thm: triangle cov hardness}
There is some $\varepsilon > 0$ such that $(\varepsilon, \varepsilon/2)$-\setpT is \np-hard.
\end{theorem}

In \cite{Petrank_1994}, they actually prove Theorem \ref{thm: triangle cov hardness} for the Max 3-Dimensional Matching problem, but Max 3-Dimensional Matching is a special case of \setpT. This correspondence is apparent by viewing each edge of a $3$-uniform hypergraph as a subset of $[n]$.
\end{paragraph}

\paragraph{Vertex cover problem.} In the minimum Vertex Cover problem (\vcov), we are given a graph $G(V,E)$ and a positive integer $k$, and the goal is to determine if there exists $C \subseteq V$ with $|C| \leq k$ such that, for all $e \in E$, at least one endpoint of $e$ is in ${C}$. Such a $C$, containing an endpoint of each edge, is called a \textit{vertex cover} of $G$.

In $(a, b)$-\vcov, we are given a graph $G(V,E)$, and the goal is to decide whether there exists a vertex cover of $G$ of size at most $a \cdot |V|$ or if every vertex cover of $G$ has size at least $b \cdot |V|$.

\section{\texorpdfstring{\apx-hardness of \cst in $\ell_\infty$-metric}{APX-hardness of CST in linf-metric}} \label{sec: l_inf cst hardness}

In this section we prove Theorem \ref{thm:introellinf}. First, we recall some inapproximability results on graph coloring.

In \cite{Cerioli_Faria_Ferreira_Martinhon_Protti_Reed_2008}, they study the problem of partitioning a cubic graph into the minimum number of independent cliques. The minimum size of such a partition is the minimum size of a partition of the complement graph into independent sets, the chromatic number of the complement graph. In particular, they show the following.
\begin{theorem}[\cite{Cerioli_Faria_Ferreira_Martinhon_Protti_Reed_2008}, Theorem 7] \label{thm: hardness of linear chromatic number}
There exist constants $\varepsilon_1 < \varepsilon_2$ such that $(\varepsilon_1 n, \varepsilon_2 n)$-coloring is \np-hard. 
\end{theorem}

Below we provide an  efficiently computable function mapping a graph $G$ of order $n$ to an instance of \cst in the $\ell_{\infty}$-metric such that the optimal cost of the Steiner tree is $(n + \chi(G))/2$. Then the \apx-hardness of \cst in the $\ell_{\infty}$-metric follows as a corollary by applying Theorem~\ref{thm: hardness of linear chromatic number}.\footnote{Theorem \ref{thm: hardness of linear chromatic number} is stated without explicit values of $\varepsilon_1$ and $\varepsilon_2$ since the current derivable parameters come from a line of works proving \apx-hardness without trying to optimize the hardness of approximation factor. If the best known approximation algorithm for partitioning cubic graphs into cliques were shown to be optimal, that would immediately imply \cst in the $\ell_{\infty}$-metric is \np-hard to approximate within a factor of $17/16$.}

Given an undirected simple graph, we arbitrarily direct/orient the edges in the graph to form a digraph $G(V, A)$. Define the function $\Gamma: V \to \R^{|A|}$ such that, for all $i \in V$,
\[
\Gamma(i) = \left(\sum_{j \,:\, (i,j) \in A}\mathbf{e_{(i,j)}}\right) -\left( \sum_{j \,:\, (j,i)\in A}\mathbf{e_{(j,i)}}\right),
\]
where $\mathbf{e_{(i,j)}}$ is the standard basis vector in $\R^{|A|}$ with $1$ in the coordinate corresponding to the directed edge $(i,j)$ and $0$ elsewhere. Let $P = \{\Gamma(i) \mid\, i \in V\}$. Our corresponding instance of \cst in the $\ell_{\infty}$-metric is on the set of terminals $\tilde{P} =  P \cup \{\mathbf{0}\}$. We will refer to $\mathbf{0}$ as the \textit{root} terminal. 

In the spirit of the language of gap-preserving reductions, we refer to the case of showing that the optimal cost of a Steiner tree of $\tilde{P}$ is at most $(n + \chi(G))/2$ as the completeness case. Similarly, the case of showing that the optimal cost of a Steiner tree of $\tilde{P}$ is at least $(n + \chi(G))/2$ is the soundness case.

\subsection{Completeness} \label{sec: l_inf completeness}
Let $\pi: V \to [a]$ be a proper $a$-coloring of the vertices in $G$. Let $A^+$ and $A^-$ be two distinct partitions of $A$ where for every  $k\in[a]$, the $k\textsuperscript{th} $ part of $A^+$ is $A^+(k):=  \{(i,j) \in A\mid \pi(i)=k\}$ and the $k\textsuperscript{th} $ part of $A^-$ is  $A^-(k) =  \{(j,i) \in A\mid \pi(i)=k\}$. Define $\sigma : [a] \to \R^{|A|}$ as follows. The function $\sigma$ will map each color class to a Steiner point. We have 
\[\forall k\in[a],\  
\sigma(k) = \left(\sum_{f \in A^+(k)} \frac{1}{2} \cdot \mathbf{e_{f}}\right) - \left(\sum_{f \in A^-(k)}\frac{1}{2} \cdot \mathbf{e_{f}}\right).
\]
Now, our set of Steiner points will be exactly the set 
$\{\sigma(k)\mid k\in [a]\}$. Note that, since $\pi$ is a proper coloring, for all $k \in [a]$, $A^+(k) \cap A^-(k) = \emptyset$. We create our Steiner tree as follows. Add an edge from each Steiner point $\sigma(k)$ to the root terminal $\mathbf{0}$. Additionally, add an edge from each non-root terminal $t \in P$ to the Steiner point $\sigma(\pi(\Gamma^{-1}(t)))$. Note that $\sigma(\pi(\Gamma^{-1}(t)))$ is the Steiner point corresponding to the color class corresponding to the vertex mapped to terminal $t$. See Figure \ref{fig: completeness construction}.

Note that the resultant Steiner tree is connected and every edge has length $1/2$. This length is clear for the edges from Steiner points to the root terminal. Now fix some $t \in P$. In particular, if $t_f = 1$, then $f \in A^+(\pi(\Gamma^{-1}(t))).$ That is, in this case, $\sigma(\pi(\Gamma^{-1}(t)))_f = 1/2$ since $\pi$ is a proper coloring. Similarly, if $t_f = -1$, then $f \in A^-(\pi(\Gamma^{-1}(t)))$. Hence, in this case, $\sigma(\pi(\Gamma^{-1}(t)))_f = -1/2$.  Note that $\norm{\sigma(\pi(\Gamma^{-1}(t)))}_{\infty} = 1/2$ by definition, so, namely, if $t_f = 0$, $|\sigma(\pi(\Gamma^{-1}(t)))_f| \leq 1/2$. Altogether, this implies that $\norm{\sigma(\pi(\Gamma^{-1}(t))) - t}_{\infty} = 1/2$.

Hence, this constructed tree has total cost $(n + a)/2$. In particular, the optimal cost of a Steiner tree in this instance of \cst in the $\ell_{\infty}$-metric is at most $(n + \chi(G))/2$.

\begin{figure}[!ht]
    \centering
   
	\includegraphics[width=\textwidth]{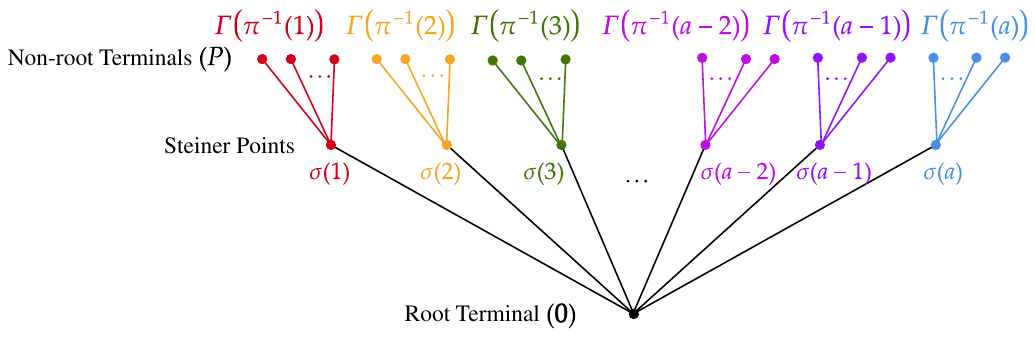}

\caption{Steiner tree of $P$ constructed from an $a$-coloring of $G$}
\label{fig: completeness construction}
\end{figure}

\subsection{Soundness}
It remains to show that the optimal cost of a Steiner tree of $\tilde{P}$ in the $\ell_{\infty}$-metric is at least $(n + \chi(G))/2$. To do so, we prove a series of lemmas which in tandem show that the Steiner tree constructed in the completeness case is in fact optimal. We emphasize that the function of these lemmas is to prove increasingly rigid structural guarantees about the optimal Steiner trees in the soundness case, building off of weaker guarantees proved in previous lemmas.

\begin{lemma} \label{lem: max dist}
Let $T(\tilde{P} \cup X, E)$ be an optimal Steiner tree of $\tilde{P}$. Then, for all $(u,v) \in E$, we have $\norm{u - v}_{\infty} \leq 1.$
\end{lemma}
\begin{proof}
Suppose $T$ has some edge of length greater than $1$. Remove the edge from $T$. Now, the resultant graph $T'$ is split into two connected components. This yields a partition of the terminals in $\tilde{P}$ into two parts. If one part contains no terminals, dropping that part yields a lower cost tree. 

Otherwise,  consider the minimum spanning tree of the terminals in $\tilde{P}$. All edges have length $1$. Namely, since it is a spanning tree, it has some edge crossing the cut induced by the partition of terminals in $T'$. Adding that edge to $T'$ forms a new Steiner tree of cost less than $T$, contradicting minimality of $T$ and proving the claim.
\end{proof}

 For ease of notation, in the remainder of this section we will make use of the following function $\sgn: \R \to \{-1, 0, 1\}$, where 
\[
\sgn(x) = \begin{cases} 
-1, &\text{if } x < 0;\\
0, &\text{if } x = 0;\\
1, &\text{if } x > 0.
\end{cases}
\]

With Lemma \ref{lem: structural props}, we begin by showing that there are optimal Steiner trees which are highly structured.

\begin{lemma} \label{lem: structural props}
There exists an optimal Steiner tree $T(\tilde{P} \cup X, E)$ with the following two properties. 
\begin{enumerate}[label=(P\arabic*)]
\item \textit{Rapid decay:} For all edges $(u,v) \in E$ with $u \in P \cup X$ and $v \in X$, if $u_i\neq 0, v_i \neq 0,$ and $\sgn(u_i) = \sgn(v_i)$, then $|u_i - v_i| = \|u - v\|_{\infty}$. \label{cond: max decrease}
\item \textit{Short edges:} Each edge $(u, v)$ for $u,v \in P \cup X$ satisfies $\|u - v\|_{\infty} < 1.$ \label{cond: short edges}
\end{enumerate}
\end{lemma}
\begin{proof}
Let $T(\tilde{P} \cup X, E)$ be an optimal Steiner tree of $\tilde{P}$. Consider the $i$\textsuperscript{th} coordinate (corresponding to some arc in $G(V,A)$). First, note that there are exactly two terminals in $P$ with nonzero $i$\textsuperscript{th} coordinate. Let $t$ be the terminal such that $t_i = 1$. Root $T$ at $t_i$ in the sense of directing each edge in $T$ away from $t_i$. Suppose that $(u,v)$ is a directed edge in $T$ such that $u_i, v_i > 0$, and $\|u - v\|_{\infty} \neq u_i - v_i$. Note that this is only possible for $u \in P \cup X$ and $v \in X$.

We construct a new Steiner tree $T'$ with $\cost_{\infty}(T') \leq \cost_{\infty}(T)$. The tree $T'$ will be the same as $T$ except for the following modifications. Set $v_i'  := \max(0, u_i - \|u - v\|_{\infty})$. Since $v_i > v_i' \geq 0$, and, for all terminals other than $t$ we have $t_i \leq 0$, this cannot increase the length of edges to terminals (compared to $T$). This is because, by Lemma \ref{lem: max dist}, since $v_i > 0$, $v$ cannot be adjacent to terminals other than $t$.
 
 Now, if for all directed edges $(v, w)$ in $T$ with $w \in X$ we have $v_i' \geq w_i$, then we have $\cost_{\infty}(T') \leq \cost_{\infty}(T)$ (since $v_i \geq v_i'$). Otherwise, suppose there exists some $w$ such that $v_i' < w_i$. Then, set $w_i'  := v_i'$. Then, if for all $(w,x)$ we have $w_i' = v_i' \geq x_i$, we have $\cost_{\infty}(T') \leq \cost_{\infty}(T)$. Otherwise, repeat this process, setting $x_i' := v_i'$, checking the $i$\textsuperscript{th} coordinate of the out-neighbors of $x'$ in $T'$ and setting their $i$\textsuperscript{th} coordinates to $v_i'$ if they are greater. Since  $T'$ is acyclic and finite, this process will eventually terminate. 
 
Repeat this process for all other directed edges $(u,v)$ with  $u_i, v_i > 0$, and $\|u - v\|_{\infty} \neq u_i - v_i$, treating the newly created optimal Steiner tree $T'$ as $T$ and the resultant tree as $T'$. Then, further repeat this process for all other coordinates $i$. Afterward, repeat this process for all coordinates $i$, except this time for all edges $(u,v)$ with $u_i, v_i < 0$ and  $\|u - v\|_{\infty} \neq  v_i - u_i$. For negative coordinates and a directed edge $(u,v)$ such that $u_i, v_i < 0$ and $\|u -v\|_{\infty} \neq v_i - u_i$, we will set $v_i^* := \min(0, u_i + \|u-v\|_{\infty})$. Since each round of modifying the tree only affects a single coordinate (and only Steiner points which are nonzero and of a given sign in those coordinates), this yields \ref{cond: max decrease}. 

Before we modify the tree further to ensure \ref{cond: short edges}, we observe that \ref{cond: max decrease} implies a useful property of the Steiner points in the Steiner tree.
\begin{claim} \label{cla: Steiner point magnitude less than 1}
Let $T(\tilde{P} \cup X, E)$ be an optimal Steiner tree satisfying \ref{cond: max decrease}. Then, for all Steiner points $s \in X$, $\norm{s}_{\infty} < 1.$
\end{claim}
\begin{subproof}
Let $t \in P$ be the unique terminal with $t_i = 1$. Then, suppose that $s \in X$ is a Steiner point adjacent to $t$. From Lemma \ref{lem: max dist}, $s_i \geq 0$. Since the distance between any two points in $T$ is greater than $0$ (assuming no degenerate Steiner points), from \ref{cond: max decrease}, $s_i < 1$. Then, consider rooting $T$ at $t$. \ref{cond: max decrease} immediately implies that the $i\textsuperscript{th} $ coordinate of any Steiner point not adjacent to $t$ is at most the maximum $i\textsuperscript{th} $ coordinate of the Steiner points adjacent to $t$. This observation holds for any coordinate $i$ and holds similarly for negatively signed coordinates, yielding the claim.
\end{subproof}

Now, suppose that $T'$ is the optimal Steiner tree resultant from the modifications for ensuring \ref{cond: max decrease}. Let $(u,v)$ be an edge with $u,v \in P \cup X$. By Lemma \ref{lem: max dist}, $\|u - v\|_{\infty} \leq 1$. If $\|u - v\|_{\infty} = 1$, consider removing the edge $(u,v)$, thereby disconnecting the tree. One of $u$ or $v$ will be in the same resultant connected component as $\mathbf{0}$. Connect the other to $\mathbf{0}$, reconnecting the tree with an edge of cost at most $1$. Repeating this process yields \ref{cond: short edges} since all edges added in each iteration of this process involve $\mathbf{0}$ and each iteration removes an edge in violation of \ref{cond: short edges}. Since this process only removes edges and adds edges adjacent to $\mathbf{0}$,  the resultant tree still satisfies \ref{cond: max decrease}.
\end{proof}

Optimal Steiner trees with the properties of Lemma \ref{lem: structural props} have some other useful properties.

\begin{corollary} \label{cor: no term term edges}
Let $T(\tilde{P} \cup X, E)$ be an optimal Steiner tree satisfying the properties of Lemma \ref{lem: structural props}. Then, for all $t,t' \in P$, $(t,t') \not\in E$.
\end{corollary}
\begin{proof}
Note that for all $t, t' \in P$ distinct, $\|t - t'\|_{\infty} \geq 1$. \ref{cond: short edges} then implies the desired result.
\end{proof}

\begin{corollary} \label{cor: independence struct}
Let $T(\tilde{P} \cup X, E)$ be an optimal Steiner tree satisfying the properties of Lemma \ref{lem: structural props}. Then, for all coordinates $i$ and  $t,t' \in P$ such that $t_i = 1$ and $t_i' = -1$, there does not exist Steiner point $s \in X$ such that $(t_i,s), (t_i', s) \in E$. That is, $t$ and $t'$ do not share a common neighboring Steiner point.
\end{corollary}
\begin{proof}
Suppose otherwise. Then, consider $s_i$. \ref{cond: short edges} implies that $s_i > 0$ and $s_i < 0$, a contradiction.
\end{proof}

We now proceed by a series of lemmas extending the structural constraints established in Lemma \ref{lem: structural props}.
\begin{lemma} \label{lem: term leaf node}
There exists an optimal Steiner tree $T(\tilde{P} \cup X, E)$ satisfying \ref{cond: max decrease}, \ref{cond: short edges}, and the following additional property:
\begin{enumerate}[label=(P\arabic*)]
\setcounter{enumi}{2}
\item \textit{Terminal connectivity:} For all $t \in P$, $t$ is a leaf node. \label{cond: term leaf node}
\end{enumerate}
\end{lemma}
\begin{proof}
Let $T(\tilde{P} \cup X, E)$ be an optimal Steiner tree satisfying \ref{cond: max decrease} and \ref{cond: short edges} (such a tree exists from Lemma \ref{lem: structural props}).  We modify $T$ without increasing the cost to yield the desired Steiner tree. 

First, by Corollary \ref{cor: no term term edges}, there are no edges between non-root terminals in $T$. Now, suppose that $t \in P$ is adjacent to $\mathbf{0}$ and some Steiner point $s$. Note that $t$ may be adjacent to other Steiner points. Drop the edge from $t$ to $\mathbf{0}$ and replace it with an edge from $\mathbf{0}$ to $s$. From Claim \ref{cla: Steiner point magnitude less than 1}, this decreases the cost of $T$, contradicting minimality. 

Finally, suppose that $t$ is adjacent to at least two Steiner points and is not adjacent to $\mathbf{0}$. Consider the subtrees of $T$ that would contain these Steiner points if we removed $t$ from the tree. In particular, one of the subtrees must contain $\mathbf{0}$ and the others cannot contain $\mathbf{0}$. Let $s$ be a Steiner point adjacent to $t$ whose resultant subtree does not contain $\mathbf{0}$. Call the subtree $T_s$. Now, drop the edge from $s$ to $t$. For each coordinate $i$ such that $t_i \neq 0$, for all Steiner points $s'$ in $T_s$, if $\sgn(s'_i) = \sgn(t_i)$, set $s'_i = 0$. This cannot increase the cost of $T_s$—the length of edges between Steiner points cannot increase, the length of edges between Steiner points and terminals with $i\textsuperscript{th} $ coordinate $0$ cannot increase, and, by \ref{cond: short edges}, any Steiner points adjacent to the terminal with opposite sign $i\textsuperscript{th} $ coordinate are unaffected.  Now add an edge  from $s$ to $\mathbf{0}$, reconnecting the tree. In particular, we now have $s_j \neq 0$ only if $t_j = 0$, so $\|s \|_{\infty}$ is at most the former distance between $s$ and $t$. Hence, these modifications did not increase the cost of the tree and removed a neighbor of $t$ without adding any neighbors to a nonzero terminal. Repeating this process—for other Steiner points $s$ adjacent to $t$ such that their subtrees resultant from removing $t$ do not contain $\mathbf{0}$—yields \ref{cond: term leaf node}. Crucially, these modifications are consistent with \ref{cond: max decrease} and \ref{cond: short edges}. After making these modifications, the tree remains optimal and we can re-apply the modifications detailed in the proof of Lemma \ref{lem: structural props} to ensure  \ref{cond: max decrease} and \ref{cond: short edges} since those modifications do not add edges to any terminals in $P$.
\end{proof}

Optimal Steiner trees with properties 
\ref{cond: max decrease}, \ref{cond: short edges}, and \ref{cond: term leaf node} reflect some the structure of the Steiner trees constructed in the completeness case. In particular, Corollary \ref{cor: independence struct} combined with \ref{cond: term leaf node} shows that these optimal Steiner trees induce a proper coloring of $V$. The color classes are defined by the (unique) Steiner points adjacent to non-root terminals, and each non-root terminal adjacent to $\mathbf{0}$ has its own color class. Given some optimal Steiner tree with properties \ref{cond: max decrease}, \ref{cond: short edges}, and \ref{cond: term leaf node}, the induced coloring of $V$ is given by the function $\pi_T$. Denote the set of non-root terminals adjacent to a Steiner point $s$ by $P(s)$ (note that $\Gamma^{-1}(P(s))$ is a color class defined by $\pi_T$).

Equipped with these structural observations, additional completeness structure emerges in the following lemma. In the statement of the lemma, $\sigma$ is defined exactly as in Section \ref{sec: l_inf completeness}, induced by the coloring $\pi_T$. In other words, for all Steiner points $s$ only adjacent to non-root terminals and $\mathbf{0}$, if $s$ is adjacent to some terminal $t$ with $t_i \neq 0$, then $s_i = t_i /2$.

\begin{lemma} \label{lem: coordinates of 1/2}
There exists an optimal Steiner tree  $T(\tilde{P} \cup X, E)$ satisfying \ref{cond: max decrease}, \ref{cond: short edges}, and \ref{cond: term leaf node}, and the following additional property.
\begin{enumerate}[label=(P\arabic*)]
\setcounter{enumi}{3}
    \item \textit{Completeness structure:}  For all Steiner points $s$ only adjacent to non-root terminals and $\mathbf{0}$, we have $s = \sigma(\pi_T(\Gamma^{-1}(P(s))))$. \label{cond: completeness structure}
\end{enumerate}
\end{lemma}
\begin{proof}
Let $T(\tilde{P} \cup X, E)$ be a Steiner tree satisfying \ref{cond: max decrease}, \ref{cond: short edges}, and \ref{cond: term leaf node}. Consider some such Steiner point $s$ in $T$ adjacent to only non-root terminals and $\mathbf{0}$. Note that, by \ref{cond: term leaf node}, since nonzero terminals are leaf nodes, $s$ must be adjacent to $\mathbf{0}$ so that $T$ is connected.

First,  we may assume that if $s_i \neq 0$, then there exists $t \in P(s)$ such that $t_i \neq 0$. Otherwise setting $s_i = 0$ does not increase the cost of the tree.

Next, if there exists $t \in P(s)$ with $t_i = 1$, then, by Corollary \ref{cor: independence struct}, $t'$ such that $t'_i = -1$ is not an element of $P(s)$. Moreover, \ref{cond: short edges} then implies that $s_i > 0$. The analogous facts are true if there exists $t \in P(s) $ with $t_i = -1$.

Now suppose that $\norm{s}_{\infty} = 1/2 + \varepsilon$ for some $\varepsilon > 0$. Then, for each $i$ such that $s_i > 0$, set $s_i = \min(s_i, 1/2)$, and, for each $j$ such that $s_j <0$, set $s_j = \max(s_j, -1/2)$. Note that, for at most one $t \in P(s)$, $\|s - t\|_{\infty}$ will increase as a result of this operation (and it will increase by at most $\varepsilon$). However, $\|s\|_{\infty}$ will decrease by $\varepsilon$, implying that the overall cost of the tree does not increase as a result of this operation (due to the edge to $\mathbf{0}$).

Finally, suppose that there is some $i$ such that $0 < |s_i| < 1/2$. Let 
\[
\min_{\{i \,:\, 0 < |s_i| < 1/2\}} |s_i| = 1/2 - \varepsilon
\]
for some $\varepsilon > 0$ and let 
\[
i_0 = \argmin_{\{i \,:\, 0 < |s_i| < 1/2\}} |s_i|.
\]
Let $\tilde{t}$ be a terminal in $P(s)$ such that $|s_{i_0}| = 1/2 - \varepsilon$ and $\tilde{t}_{i_0} \neq 0$. Then, for each $i$ such that $s_i > 0$, set $s_i = 1/2$ and for each $j$ such that $s_j < 0$, set $s_j = -1/2$. Now, this modification decreases the distance from $s$ to $\tilde{t}$ by $\varepsilon$ while increasing the distance from $\mathbf{0}$ to $s$ by at most $\varepsilon$. Moreover, for all $t \in P(s)$, this does not increase $\|t - s\|_{\infty}$ since, from the previous modification, $\|t - s\|_{\infty} \geq 1/2$. Hence, this does not increase the cost of $T$ and yields \ref{cond: completeness structure}. Note that \ref{cond: max decrease}, \ref{cond: short edges}, and \ref{cond: term leaf node} are all preserved under these operations.
\end{proof}

Lemma \ref{lem: coordinates of 1/2} shows that when Steiner points are not adjacent to other Steiner points, they appear exactly like Steiner points in the completeness case. Our next objective is then to show that we may assume there are no edges between Steiner points. To do so, we will first consider the most ``extremal'' Steiner points which are not accounted for by Property \ref{cond: completeness structure}. Let $T(\tilde{P} \cup X, E)$ be an optimal Steiner tree satisfying \ref{cond: max decrease}, \ref{cond: short edges}, \ref{cond: term leaf node}, and \ref{cond: completeness structure}. Then, removing all terminals in $P$ will result in a smaller (connected) tree since all such terminals are leaf nodes by \ref{cond: term leaf node}. Call the non-terminal leaf nodes in the resultant tree \textit{Steiner leaves}. We show that Steiner leaves are highly structured in Lemma \ref{lem: Steiner leaf struct}.

\begin{lemma} \label{lem: Steiner leaf struct}
There exists an optimal Steiner tree  $T(\tilde{P} \cup X, E)$ satisfying \ref{cond: max decrease}, \ref{cond: short edges}, \ref{cond: term leaf node}, \ref{cond: completeness structure}, and the following additional property.

\begin{enumerate}[label=(P\arabic*)]
\setcounter{enumi}{4}
    \item \textit{Steiner leaf structure:}  For all Steiner leaves $s \in X$, $\|s\|_{\infty} = 1/2$ \label{cond: Steiner leaf structure}. Moreover, for each coordinate $i$ such that there exists $t \in P(s)$ with $t_i \neq 0$, $|s_i| = 1/2$.
\end{enumerate}
\end{lemma}
\begin{proof}
Let $T(\tilde{P} \cup X, E)$ be a Steiner tree satisfying \ref{cond: max decrease}, \ref{cond: short edges}, \ref{cond: term leaf node}, and \ref{cond: completeness structure}. Suppose that  $s$ is a Steiner leaf in $T$ not satisfying \ref{cond: Steiner leaf structure}. Using the triangle inequality, we may assume that $s$ has degree at least $3$ (otherwise, the Steiner point may be replaced by an edge directly between its two neighbors without increasing the cost of the tree). Such an application of the triangle inequality is consistent with \ref{cond: max decrease}, \ref{cond: term leaf node}, and \ref{cond: completeness structure}. The argument in Lemma \ref{lem: structural props} to achieve \ref{cond: short edges} can be reapplied if this results in an edge of length $1$. Then, the fact that $s$ is a Steiner leaf implies that $|P(s)| \geq 2$. If $s$ is only adjacent to terminals, we may apply \ref{cond: completeness structure}.

Otherwise, $s$ must be adjacent to one Steiner point $s'$ and at least two nonzero terminals. Suppose that $\norm{s}_{\infty} = 1/2 + \varepsilon$ for some $\varepsilon > 0$. There are two cases to consider. 

First, suppose that for all $i$ such that $|s_i| = 1/2 + \varepsilon$, there exists $t \in P(s)$ such that $\sgn(t_i) = \sgn(s_i)$. Let $\gamma = \max(\{|s_i| \,:\, \forall t \in P(s), t_i = 0\} \cup \{ 1/2\}).$ That is, $\gamma$ is a maximum magnitude of a coordinate of $s$ not corresponding to a coordinate of a terminal in $P(s)$ (or $1/2$ if no such coordinates are greater than $1/2$). Now, for all coordinates $i$ such that $t_i = 1/2 + \varepsilon$, set $t_i = \gamma$. Likewise, for all coordinates $j$ such that $t_j = - 1/2 - \varepsilon$, set $t_j = -\gamma$. This increases $\| s - t\|_{\infty}$ for at most one $t \in N(s)$ by at most $1/2 + \varepsilon - \gamma$, but for all other $t' \in N(s)$, $\| s - t'\|_{\infty}$ decreases by $1/2 + \varepsilon - \gamma$. 
Hence, this can only increase the cost of tree if the tree if $\|s - s'\|_{\infty}$ increases. But, we may propagate this reduction in magnitude of coordinates through $s'$ as in Lemma \ref{lem: structural props} in order to reestablish \ref{cond: max decrease} and ensure that this change does not increase $\|s - s'\|_{\infty}$, thereby not increasing the cost of the tree. Note that this relies on our assumption that, for all $i$ such that $|s_i| = 1/2 + \varepsilon$, there exists $t \in P(s)$ such that $\sgn(t_i) = \sgn(s_i)$. Afterward, we may apply the modifications outlined in Lemma \ref{lem: structural props} to re-establish \ref{cond: max decrease} and \ref{cond: short edges} since they will not increase the magnitude of coordinates of Steiner leaves.

Now, if we still have  $\|s\|_{\infty} = 1/2 + \varepsilon' > 1/2$, there exists some coordinate $i$ such that $|s_i| = \|s\|_{\infty}$ and for all $t \in P(s)$, $t_i = 0$. Now, for coordinates $j$ such that $s_j > 1/2$, set $s_j = 1/2$ and, likewise, if $s_j < -1/2,$ set $s_j = -1/2$. Then, for all $t \in P(s)$, $\|s - t\|_{\infty}$ decreases by $\varepsilon'$ and $\|s - s'\|_{\infty}$ increases by at most $\varepsilon'$. Since $|P(s)| \geq 2$, this actually decreases the cost of the tree, contradicting minimality. Hence, this case cannot occur.

It remains to show that, for each coordinate $i$ such that there exists $t \in P(s)$ with $t_i \neq 0$, $|s_i| = 1/2$. Call such coordinates \textit{corresponding} coordinates of $s$ (with the terminal $t \in P(s)$ corresponding to such a coordinate referred to as a \textit{corresponding terminal}). Let $i$ be a corresponding coordinate of $s$ such that $|s_i|$ is minimal; say $|s_i| = 1/2 - \varepsilon$. By \ref{cond: short edges}, $\sgn(s_i) = \sgn(t_i)$ for corresponding terminal $t.$ By the above, $\|s\|_{\infty} \leq 1/2$. Now, increasing the magnitude of all corresponding coordinates of $s$ (while retaining their sign) will decrease the distance to $t$ by $\varepsilon$ and increase the distance from $s$ to $s'$ by at most $\varepsilon$. Hence, repeating this operation does not increase the cost of the Steiner tree while ensuring Property \ref{cond: Steiner leaf structure}.  If \ref{cond: max decrease} or \ref{cond: short edges} are ever violated, they be restored by the processes described in Lemma \ref{lem: structural props} without violating \ref{cond: Steiner leaf structure}. Properties \ref{cond: term leaf node} and \ref{cond: completeness structure} are unaffected by the modifications outlined above.
\end{proof}

The language introduced in Lemma \ref{lem: Steiner leaf struct} will be useful in the following lemmas. For $s$ a Steiner point, call $C(s) = \{i \,:\, t \in P(s), t_i \neq 0\}$ the set of corresponding coordinates of $s$. We observe an implication of Lemma \ref{lem: Steiner leaf struct}.

\begin{lemma} \label{lem: Steiner leaf independent neigbors}
There exists a optimal Steiner  $T(\tilde{P} \cup X, E)$ satisfying \ref{cond: max decrease}, \ref{cond: short edges}, \ref{cond: term leaf node}, \ref{cond: completeness structure}, \ref{cond: Steiner leaf structure}, and the following additional property:
\begin{enumerate}[label=(P\arabic*)]
\setcounter{enumi}{5}
    \item \textit{Steiner leaf independence:}
    Let $s$ be a Steiner leaf. Then, if $s''$ is a Steiner point adjacent to $s$, then $C(s) \cap C(s'') = \emptyset$. Moreover, for all other Steiner leaves $s'$ adjacent to $s''$, $C(s) \cap C(s') = \emptyset$. \label{cond: Steiner leaf independence}
\end{enumerate}
\end{lemma}
\begin{proof}
Let $T(\tilde{P} \cup X, E)$ be a Steiner tree satisfying \ref{cond: max decrease}, \ref{cond: short edges}, \ref{cond: term leaf node},  \ref{cond: completeness structure}, and \ref{cond: Steiner leaf structure}. Let $s$ be a Steiner leaf and $s''$ be a Steiner point neighboring $s$. First, suppose that $i \in C(s) \cap C(s'')$. Then, $s_i = 1/2$ (using \ref{cond: Steiner leaf structure}). Note that $s''_i < 0$ by \ref{cond: short edges}. Then, dropping the edge $(s, s'')$ and adding the edge $(s, \mathbf{0})$ maintains the connectivity of the tree while decreasing its cost, contradicting optimality of $T$.

Now, let $s'$ be a second Steiner leaf neighboring $s''$ and let $i \in C(s) \cap C(s')$. Without loss of generality assume $s_i = 1/2$ and $s_i' = -1/2$. Then, clearly either $\|s'' - s\|_{\infty} \geq 1/2$ or $\|s'' - s'\|_{\infty} \geq 1/2$. Assume that $\|s'' - s\|_{\infty} \geq 1/2$. Removing edge $(s, s'')$ and adding edge $(s, \mathbf{0})$ then removes this pair of Steiner points violating \ref{cond: Steiner leaf independence} while keeping the tree connected and not increasing its cost.  The fact that the cost does not increase follows from \ref{cond: Steiner leaf structure}. This modification may affect \ref{cond: max decrease}, \ref{cond: completeness structure}, \ref{cond: Steiner leaf structure}, but following the algorithms described in Lemmas \ref{lem: structural props}, and \ref{lem: coordinates of 1/2} and \ref{lem: Steiner leaf struct} can re-establish these properties without increasing the number of pairs of Steiner leaves violating \ref{cond: Steiner leaf independence}. Repeating this process then yields the result.
\end{proof}

We are now ready to prove the final and most important structural property of optimal Steiner trees: there exist optimal Steiner trees with no edges between Steiner points. 
\begin{lemma} 
There exists a optimal Steiner  $T(\tilde{P} \cup X, E)$ satisfying \ref{cond: max decrease}, \ref{cond: short edges}, \ref{cond: term leaf node}, \ref{cond: completeness structure}, \ref{cond: Steiner leaf structure}, \ref{cond: Steiner leaf independence}, and the following additional property:
\begin{enumerate}[label=(P\arabic*)]
\setcounter{enumi}{6}
    \item \textit{Steiner adjacency:} There are no edges between Steiner points in $T.$ \label{cond: no Steiner edges}
\end{enumerate}
\end{lemma}
\begin{proof}
Let $T(\tilde{P} \cup X, E)$ be an optimal Steiner tree satisfying \ref{cond: max decrease}, \ref{cond: short edges}, \ref{cond: term leaf node}, \ref{cond: completeness structure}, \ref{cond: Steiner leaf structure}, and \ref{cond: Steiner leaf independence}. To prove this result, we begin with an important claim.
\begin{claim}
We may assume that, for all non-leaf $s \in X$, $s$ is adjacent (in $T$) to $\mathbf{0}$ and at least one other non-leaf Steiner point or at least two non-leaf Steiner points.
\end{claim}
\begin{subproof}
Suppose some non-leaf Steiner point $s$ in $T$ does not have this property. Since $s$ is not a Steiner leaf, it must be adjacent to at least two points that are not nonzero terminals. Hence, it is adjacent to at least one Steiner leaf $s'$. By \ref{cond: Steiner leaf independence}, we have $C(s) \cap C(s') = \emptyset$ and, for any two Steiner leaves $s'$ and $s''$ adjacent to $s$,  $s'$ and $s''$, $C(s') \cap C(s'') = \emptyset$. 

Note that by \ref{cond: max decrease}, all coordinates in $C(s')$ behave the same in the following sense. Upon rooting the tree at $s'$ and directing all edges away from $s'$, these coordinates decrease maximally in magnitude until reaching $0$. So, until they reach $0$, they all have the same magnitude. Now, fix a Steiner leaf $s'$ adjacent to $s$. Remove all edges from $s$ to its neighboring nonzero terminals, add edges from those nonzero terminals to $s'$, and treat each coordinate previously in $C(s)$ like the coordinates previously in $C(s')$ (setting them to have the same magnitudes and appropriate signs), propagating through the tree until the coordinates reach $0$. This does not increase the cost of the tree since $s'$ is a Steiner leaf, the tree satisfies property \ref{cond: max decrease}, and we had $C(s) \cap C(s') = \emptyset$.

Do the same process outlined above for all other Steiner leaves adjacent to $s$ (replace their edges to nonzero terminals with edges from those terminals to $s'$ and have their corresponding coordinates emulate the coordinates of other nonzero terminals adjacent to $s'$). Also drop those Steiner leaf Steiner points.

After this process, $s$ has degree $2$, so, by the triangle inequality, it can be removed and an edge can be added directly between its two neighbors. We can repeat this process until $s'$ is either adjacent to $\mathbf{0}$ or a non-leaf Steiner point satisfying the desired property. At some points, we may need to restore \ref{cond: Steiner leaf independence}, but we can do so by following the algorithm described in Lemma \ref{lem: Steiner leaf independent neigbors}.

Repeating for all non-leaf Steiner points without this property then ensures the property holds in the resultant tree. After each step, we can restore properties \ref{cond: max decrease} and \ref{cond: short edges} by applying the algorithm described in Lemma \ref{lem: structural props}. The other properties are unaffected by these modifications.
\end{subproof}
Given the claim, consider removing non-root terminals and Steiner leaf nodes from $T$. By the claim, all remaining Steiner points are degree at least two. But then, by \ref{cond: term leaf node}, the only possible leaf node in the resultant graph is $\mathbf{0}$. The only tree with only one leaf node is a singleton vertex, implying all Steiner points in $T$ are Steiner leaves. Hence, we have \ref{cond: no Steiner edges}.
\end{proof}

Lemma \ref{lem: coordinates of 1/2} in fact show that there exist optimal Steiner trees with exactly the structure of the completeness case. This is formalized in Lemma \ref{lem: soundness final lemma}.
\begin{lemma} \label{lem: soundness final lemma}
There exists an optimal Steiner tree in which the length of each edge is $1/2$ and the number of Steiner points is exactly $\chi(G)$.
\end{lemma}
\begin{proof}
Let $T(\tilde{P} \cup X, E)$ be an optimal Steiner tree satisfying properties \ref{cond: max decrease}, \ref{cond: short edges}, \ref{cond: term leaf node}, \ref{cond: completeness structure}, \ref{cond: Steiner leaf structure}, \ref{cond: Steiner leaf independence}, and \ref{cond: no Steiner edges}. 

By \ref{cond: term leaf node}, \ref{cond: no Steiner edges}, \ref{cond: short edges}, and \ref{cond: completeness structure}, each terminal is either connected to a Steiner point by an edge of length $1/2$ or connected to $\mathbf{0}$ by an edge of length $1$. Each edge from a Steiner point to $\mathbf{0}$ is length $1/2$. 

 For each edge of length $1$ from a terminal, we can add an intermediate Steiner point halfway along that edge (corresponding to the terminal) such that the edge is split into two edges of length $1/2$. Then, the cost is fixed and each edge is of length $1/2$.

Now, by Corollary \ref{cor: independence struct}, for each Steiner point $s$, $P(s)$ corresponds to an independent set in $G$. Since each edge in the graph is length $1/2$, the cost of the tree is $(n + |X|)/2$, where $X$ is the set of Steiner points in the tree. Since each nonzero terminal is adjacent to a Steiner point, by Corollary \ref{cor: independence struct}, $T$ must have at least $\chi(G)$ Steiner points. But, the construction in the completeness case shows a tree with $\chi(G)$ Steiner points and this structure is possible (adding the intermediate Steiner points for edges directly to $\mathbf{0}$), implying the result.
\end{proof}

This completes the proof of the soundness case of Theorem \ref{thm:introellinf}.

\section{The Metric Steiner Problem on Graphs}\label{sec:metricspace}
In this section, we show a gap preserving reduction from \setpT to \dst in general metrics. In Section \ref{subsec: metric from sets} we describe how we associate a metric space to a set system. The resultant metric spaces will be called \textit{set system spaces}. When the distances in these metric spaces satisfy an additional collection of constraints (see Definition \ref{defn: Steiner embeddability}), these set system spaces will be sufficient to yield a gap preserving reduction from \setpT to \dst. The details of this reduction are provided in Subsection \ref{subsec: gap preserving red metric st}.

\subsection{From Set Systems to Metric Spaces} \label{subsec: metric from sets}
In this subsection, we describe how to associate a  metric space to a set system\footnote{For the purposes of our application to \dst, we specify this association to metric spaces only for  set systems where each set in the collection is of size exactly 3.}. We also introduce more restricted metric spaces that facilitate gap preserving reductions from \setpT to \dst.

\begin{definition}[Set System Space] \label{defn: set system space}
Let $([n], \mathcal{S})$ be a set system such that for all $S \in \mathcal{S}$, we have $|S| = 3$, and $|\mathcal{S}|:= m$. Let $\srd, \trd, \beta_{\text{in}}, \beta_{\text{out}}, \gamma_0, \gamma_1, \gamma_2, \td \in \R^+$. 
Then, we define the $(\srd, \trd, \beta_{\text{in}}, \beta_{\text{out}}, \gamma_0, \gamma_1, \gamma_2, \td)$-\textit{set system space} corresponding to $([n], \mathcal{S})$ to be the space  $(\tilde{P} \cup X, \Delta)$ containing $n + m + 1$ points and a distance function $\Delta$. We define $P := \{t_i \,:\, i \in [n]\}$, $\tilde{P} = P \cup \{r\}$, and $X := \{s_j \,:\, S_j \in \mathcal{S}\}$. We refer to the points in $P$ as \textit{universe elements}, $r$ as the \textit{root}, and the points in $X$ as \textit{facilities}. For convenience, for $s_j \in X$, we define $\Gamma(s_j) = S_j$.

Interpreting the points of the space as vertices in a graph, we define the weight function $\Delta$ by defining every pairwise edge weight as follows (the weight of the edge between each point and itself is $0$).
\begin{itemize}
    \item $\Delta(r, s) = \srd$ for all $s \in X$.
    \item $\Delta(r,t) = \trd $ for all $t \in P$.
    \item $\Delta(t_i, s) = \beta_{\text{in}}$ for all $i \in [n]$ such that $i \in \Gamma(s)$.
    \item $\Delta (t_i, s) = \beta_{\text{out}} $ for all $i \in [n]$ such that $i \not\in \Gamma(s)$.
    \item $\Delta(s, s') = \gamma_i$ for all $s, s' \in X$ such that $|\Gamma(s) \cap \Gamma(s')| = i$, for $i = 0,1,$ and $2$.
    \item $\Delta(t,t') = \td$ for all distinct $t,t' \in P$.
\end{itemize}
\end{definition}

To elucidate our choice of notation, we also provide intuitive definitions of each parameter. 
\begin{itemize}
\item $\srd$ is the distance from the root terminal to a Steiner point in $X$.
    \item $\trd$ is the distance from the root terminal to a terminal in $P$.
    \item $\beta_{\text{in}}$ is the distance between a Steiner point $s$ and non-root terminal $t_i$ such that the set corresponding to $s$ contains the universe element corresponding to $t_i$. 
    \item $\beta_{\text{out}}$ is the distance between a Steiner point $s$ and non-root terminal $t_i$ such that the set corresponding to $s$ \textit{does not} contain the universe element corresponding to $t_i$.
    \item $\gamma_i$ is the distance between Steiner points $s$ and $s'$ such that their corresponding sets have an intersection of size $i$.
    \item $\tau$ is the distance between non-root terminals.
\end{itemize}

Note that, although we say ``distance'' in the above intuitive descriptions, the weight function $\Delta$ may not define a metric, so set system spaces may not be metric spaces. The following definition resolves this issue.

\begin{definition}[Metric Compatibility] \label{defn: metric compat}
 Suppose that we have $\srd, \trd, \beta_{\text{in}}, \beta_{\text{out}}, \gamma_0, \gamma_1, \gamma_2, \td \in \R^+$, such that all of the following inequalities hold.
\begin{enumerate}
    \item $\srd \leq \min(\trd + \beta_{\text{in}},\trd + \beta_{\text{out}})$.
    \item $\trd \leq \min(\srd + \beta_{\text{in}}, \srd + \beta_{\text{out}}).$
    \item $\beta_{\text{in}} \leq  \min(\srd + \trd,$ $\beta_{\text{out}} + \td, \gamma_i + \beta_{\text{out}}$ \,:\, $i \in \{0,1,2\})$.
    \item $\beta_{\text{out}} \leq \min(\srd + \trd, \beta_{\text{in}} + \td, \gamma_i + \beta_{\text{in}} \,:\, i \in \{0,1,2\})$.
    \item $\gamma_i \leq \min(2 \srd, \gamma_j + \gamma_k, 2\beta_{\text{in}}, 2\beta_{\text{out}}$ \,:\, $i,j,k \in \{0,1,2\})$.
    \item $\td \leq \min(2 \trd$, $2\beta_{\text{in}},$ $2\beta_{\text{out}})$.
\end{enumerate}
In this case, we call the tuple $(\srd, \trd, \beta_{\text{in}}, \beta_{\text{out}}, \gamma_0, \gamma_1, \gamma_2, \td)$ \textit{metric compatible}. Note that the constraints $\gamma_i \leq \beta_{\text{in}} + \beta_{\text{out}}$ and $\td \leq \beta_{\text{in}} + \beta_{\text{out}}$ are implied by $\gamma_i \leq 2\beta_{\text{in}},  2\beta_{\text{out}}$ and $\td \leq 2\beta_{\text{in}}, 2\beta_{\text{out}}$. 
\end{definition}

A general diagram of an $(\srd, \trd, \beta_{\text{in}}, \beta_{\text{out}}, \gamma_0, \gamma_1, \gamma_2, \td)$-set system space is given in Figure \ref{fig: set system space}. Note that not all pairwise distances in the diagram are labeled.

\begin{figure}[!ht]
    \centering

	\includegraphics[width=0.7\textwidth]{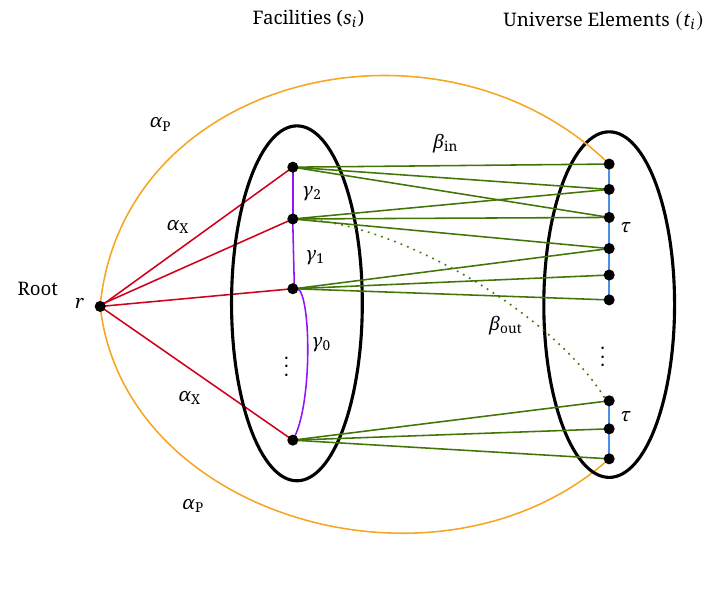}

    \caption{A general $(\srd, \trd, \beta_{\text{in}}, \beta_{\text{out}}, \gamma_0, \gamma_1, \gamma_2, \td)$-set system space with edges between different groups of nodes color-coded}
    \label{fig: set system space}
\end{figure}

\begin{proposition}
If $(\srd, \trd, \beta_{\text{in}}, \beta_{\text{out}}, \gamma_0, \gamma_1, \gamma_2, \td)$ is metric compatible, then any \\$(\srd, \trd, \beta_{\text{in}}, \beta_{\text{out}}, \gamma_0, \gamma_1, \gamma_2, \td)$-set system space is a metric space.
\end{proposition}
\begin{proof}
Inequalities 1–6 in Definition \ref{defn: metric compat} are precisely those necessary for the triangle inequality to hold with distance function $\Delta$. The other metric constraints follow trivially from Definition \ref{defn: set system space}.
\end{proof}

Hence, if $(\srd, \trd, \beta_{\text{in}}, \beta_{\text{out}}, \gamma_0, \gamma_1, \gamma_2, \td)$ is metric compatible, associated set system space weight functions $\Delta$ are distance functions. Now we introduce the notion of \textit{Steiner embeddability}, a further restriction on metric compatible tuples. Tuples with these restrictions induce a notion of distance on metric spaces related to set systems that facilitates the construction of a gap preserving reduction from \setpT to \dst.

\begin{definition}[Steiner embeddability and Steiner spaces] \label{defn: Steiner embeddability}
We call a metric compatible tuple \\ $(\srd, \trd, \beta_{\text{in}}, \beta_{\text{out}}, \gamma_0, \gamma_1, \gamma_2, \td)$ \textit{Steiner embeddable} if 
\begin{enumerate}[label=(P\arabic*)]
\item \textit{Steiner proximity:} $\srd \leq 3\gamma_2/2,$ $\trd,$ $\beta_{\text{in}},$ $\beta_{\text{out}},$ $\gamma_0$, $\gamma_1,$ $\td$. \label{cond: sr prox}
    \item \textit{Root proximity:} $\trd \leq \beta_{\text{out}}.$  \label{cond: res s adj}
    \item \textit{Steiner utility:} $\beta_{\text{in}} + \frac{\srd}{3} < \min(\trd, \td)$.  \label{cond: s util}
    \item \textit{Steiner diameter:} $\min(\trd, \td) \leq \beta_{\text{in}} + \gamma_2$.  \label{cond: ss' non adj}
\end{enumerate}

We call any $(\srd, \trd, \beta_{\text{in}}, \beta_{\text{out}}, \gamma_0, \gamma_1, \gamma_2, \td)$-set system space with $(\srd, \trd, \beta_{\text{in}}, \beta_{\text{out}}, \gamma_0, \gamma_1, \gamma_2, \td)$ Steiner embeddable a \textit{Steiner space}.  The constraints in Definition \ref{defn: metric compat} ensure that the parameters can be realized as distances in a metric space, and the constraints in Definition \ref{defn: Steiner embeddability} ensure sufficient properties for proving hardness of approximation of \dst via a reduction from $(\varepsilon, \delta)$-\setpT.
\end{definition}

\subsection{Hardness of \texorpdfstring{\dst}{DST} from Set Packing} \label{subsec: gap preserving red metric st}

 In this subsection we describe a general reduction from $(\varepsilon,\delta)$-\setpT to \dst using the language of Steiner spaces.

\begin{theorem} \label{thm: graph hardness}
Let $(\srd,$ $\trd,$ $\beta_{\text{in}},$ $\beta_{\text{out}},$ $\gamma_0,$ $\gamma_1,$ $\gamma_2,\td)$ be a Steiner embeddable tuple.
For an instance  $([n],\mathcal{S})$ of $(\varepsilon,\delta)$-\setpT where $|\mathcal{S}|=m$, let $(\tilde{P}\cup X,\Delta)$ be the Steiner space guaranteed by Definitions~\ref{defn: metric compat} and \ref{defn: Steiner embeddability}. Then,  we have the following guarantees on an instance $(\tilde{P} \cup X)$ of \dst over $(\tilde{P}\cup X,\Delta)$:
\begin{description}
        \item[Completeness:] If $([n],\mathcal{S})$ admits a set packing of size $n/3$, then there is a Steiner tree for $\tilde{P}$ of cost $$n(\srd/3+\beta_{\text{in}}).$$
        \item[Soundness:] Every Steiner tree of $\tilde{P}$ must be of cost at least
        $$n(\srd/3+\beta_{\text{in}})\cdot \left(1 - \varepsilon  + \frac{(4\varepsilon  - 2\delta)}{3} \cdot \frac{\min(\trd, \td, \beta_{\text{in}} + \srd/2)}{\srd/3+\beta_{\text{in}}} + \frac{(2\delta  - \varepsilon)}{3} \cdot \frac{\min(\trd, \td)}{\srd/3+\beta_{\text{in}}} \right).$$

        \item[Run Time:] Given $\srd,$ $\trd,$ $\beta_{\text{in}},$ $\beta_{\text{out}},$ $\gamma_0,$ $\gamma_1,$ $\gamma_2,\td$ and $([n],\mathcal{S})$, the above instance of \dst can be constructed in $\poly(n+m)$ time. 
        \end{description}
\end{theorem}

\begin{remark}
The completeness case bound comes from Steiner trees of the form shown in Figure \ref{fig: DST completeness graph}. In the soundness case, we prove that the optimal structure of the Steiner tree involves packing as many terminals as possible into a tree structure analogous to the completeness case: groups of three terminals each at distance $\beta_{\text{in}}$ from a Steiner point connected to the root $r$ (we conclude this by Lemma \ref{lem: graph max packing}). This is the $(1-\varepsilon)$ term. Then, whether its better to try and group the remaining terminals into pairs connected to common Steiner points and handle them individually depends precisely on which term minimizes $\min(\trd, \td, \beta_{\text{in}} + \srd/2)$. The first term is the per-terminal cost of connecting the remaining terminals each indvidually to the root, the second term is the per-terminal cost of connecting the remaining terminals each to some terminal already connected (via a Steiner point) to the root, and the final term in the per-terminal cost of connecting paired up terminals to the root via a common Steiner point. If $\beta_{\text{in}} + \srd/2$ minimizes this expression, then it is optimal to pair up as many terminals as possible. However, the number of achievable pairs is limited by the minimum size set cover restriction in the soundness case of \setpT (this is where the $(4\varepsilon - 2\delta)/3$ comes from). Then, the remaining terminals are connected to the root either directly or via edges to already connected terminals.
\end{remark}

Nonetheless, the expression for the inapproximability factor in the soundness case is quite intricate. One may wonder if this theorem implies a stronger inapproximability result than the simple reduction from Set Cover mentioned in Section~\ref{sec:techniques} for \dst in general metrics.  
\begin{corollary} \label{cor: general metric useless?}
Suppose $(\varepsilon, \delta)$-\setpT is \np-hard and there exists a Steiner embeddable tuple $(\srd,$ $\trd,$ $\beta_{\text{in}},$ $\beta_{\text{out}},$ $\gamma_0,$ $\gamma_1,$ $\gamma_2,\td)$ that also satisfies 
\begin{enumerate}
    \item $\srd = \beta_{\text{in}}$,
    \item $\min(\trd, \td, \beta_{\text{in}} + \srd/2) = 3\srd/2,$ and
    \item $\min(\trd, \td) =2\srd$.
\end{enumerate}
Then, it is \np-hard to approximate \dst in general metric spaces within a factor of $1 + \delta/4$.
\end{corollary}
\begin{proof}
The hardness of approximation factor is achieved by dividing the costs in the soundness and completeness cases of Theorem \ref{thm: graph hardness}. Then, we wish to maximize both
\[
\frac{\min(\trd, \td, \beta_{\text{in}} + \srd/2)}{\srd/3+\beta_{\text{in}}} 
\]
and 
\[
\frac{\min(\trd, \td)}{\srd/3+\beta_{\text{in}}}.
\]
The numerator of the first fraction is at most $\beta_{\text{in}} + \srd/2$. The numerator of the second fraction is at most $\trd \leq \srd + \beta_{\text{in}}$ by the triangle inequality. Supposing that both inequalities were in fact tight, our problem would reduce to maximizing  $\frac{\srd/2 + \beta_{\text{in}} }{\srd/3+\beta_{\text{in}}}$ and $\frac{\srd + \beta_{\text{in}}}{\srd/3+\beta_{\text{in}}}$. Both are maximized for minimum $\beta_{\text{in}}$ and, since $\srd \leq \beta_{\text{in}}$ by \ref{cond: sr prox}, the best possible is $\beta_{\text{in}} = \srd$. Assuming there is some tuple $(\srd,$ $\trd,$ $\beta_{\text{in}},$ $\beta_{\text{out}},$ $\gamma_0,$ $\gamma_1,$ $\gamma_2,\td)$ that both satisfies these additional constraints and is Steiner embeddable, we would then have that \dst is \np-hard to approximate within a factor of 
\[
1 - \varepsilon + \frac{(4\varepsilon  - 2\delta)}{3}\cdot\frac{9}{8} + \frac{(2\delta  - \varepsilon)}{3} \cdot \frac{3}{2} = 1 + \delta/4.\qedhere
\]
\end{proof}
In fact, \dst in the $\ell_{\infty}$-metric can satisfy these constraints and achieve this hardness of approximation factor. See Theorem \ref{thm: graph l_inf hardness}.
\begin{remark}
Interestingly, as shown in Section \ref{sec:techniques}, there exists a simple reduction from Set Cover with sets of size $3$ to \dst which yields this same hardness in general metric spaces. Nonetheless, $\varepsilon$, the set packing parameter in $(\varepsilon, \delta)$-\setpT is highly relevant in certain other metric spaces. See, for example Theorem \ref{thm: hardness of lp dst} and Corollary \ref{cor: l_p tight hardness} for applications to $\ell_p$ metric spaces.
\end{remark}

We now prove Theorem \ref{thm: graph hardness}. Given any hard instance of \setpT, using the language of Subsection \ref{subsec: metric from sets}, we can consider the corresponding $(\srd,$ $\trd,$ $\beta_{\text{in}},$ $\beta_{\text{out}},$ $\gamma_0,$ $\gamma_1,$ $\gamma_2,\td)$-set system space $((P \cup \{r\})\cup X, \Delta)$. This space may be interpreted as an instance of \dst with terminals $\tilde{P} = P \cup \{r\}$, candidate Steiner points $X$, and pairwise distances given by $\Delta$ determining edge weights.

 \begin{paragraph}{\textit{Completeness.}}
 In the completeness case, the \setpT instance has a solution that partitions $[n]$ into sets from $\mathcal{S}$. Let $C$ be such a partition. For each $\{i,j,k\} \in C$, connect $t_i$, $t_j$, and $t_k$ to $s = \Gamma^{-1}(\{i,j,k\})$. Then, connect each such $s$ to $r$. This results in a Steiner tree of total cost $n(\srd/3 + \beta_{\text{in}})$. See Figure \ref{fig: DST completeness graph}.
 
 \begin{figure}[!ht]
     \centering
     
	\includegraphics[width=\textwidth]{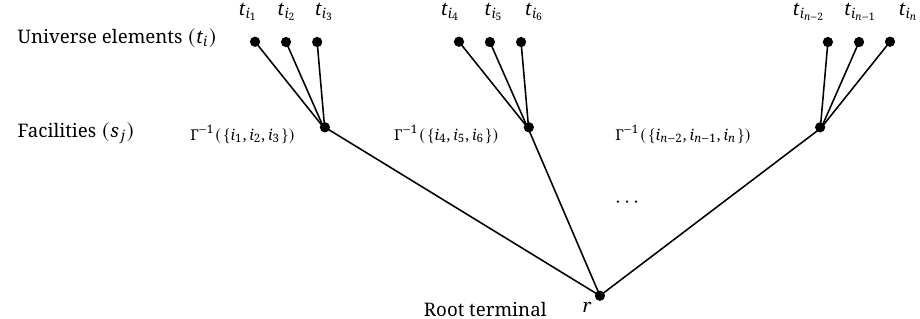}

     \caption{Steiner tree construction in the completeness case of Theorem \ref{thm: graph hardness}.}
     \label{fig: DST completeness graph}
 \end{figure}
 \end{paragraph}

 \begin{paragraph}{\textit{Soundness.}}
 In the soundness case, we prove a series of consistent claims about the adjacencies of universe elements (non-$r$ terminals) and Steiner points. The claims rely intimately on the properties of Steiner embeddable tuples. Ultimately, the claims will show that there is some minimum length Steiner tree mirroring much of the structure of the completeness case. This then yields insight into the size of a set packing in the \setpT instance, thereby lower-bounding the length of a minimum Steiner tree.

We want $r$ to serve as a root for the Steiner tree with its subtrees inducing groups of related Steiner points. To achieve this goal, we need to restrict the adjacencies of both universe elements and Steiner points. In Proposition \ref{prop: graph term st edges} we show that there exists a minimum length Steiner tree such that, for each universe element $t_i$, $t_i$ is adjacent to only Steiner points $s$ such that $i \in \Gamma(s)$. We hope to interpret the Steiner points in the Steiner tree as a choice of sets from $\mathcal{S}$. This property is important since it means the adjacency of a universe element to a Steiner point can be interpreted as the set corresponding to the Steiner point covering the universe element.

Lemma \ref{lem: graph deg condition term} further restricts the adjacencies of universe elements. It shows that we may assume that edges between universe elements cannot connect substantial portions of the tree—one of the endpoints of such edges must be a leaf node. Indeed, the leaf node endpoint will ultimately become a universe element not covered by the set packing induced by the Steiner tree. Notably, Lemma \ref{lem: graph deg condition term} leaves open the option that a universe element could serve as a connecting hub for the tree, avoiding the need for Steiner points and limiting the role of $r$. Lemma \ref{lem: graph term leaf node}, proved using Lemma \ref{lem: graph st existence}, discounts that possibility, enforcing that non-leaf node universe elements must be connected to Steiner points.  However, even with Lemma \ref{lem: graph deg condition term}, terminals could conceivably connect to many Steiner points adjacent to other terminals. Lemma \ref{lem: graph terminals only adj to one st} shows that we may assume that this is not the case. Proposition \ref{prop: graph terminals not adj to r if in triangle}, combined with Lemma \ref{lem: graph term leaf node}, shows that we may assume that universe elements of degree greater than $1$ are not adjacent to $r$. This yields some insight into how groups of terminals connected to common Steiner points are connected to $r$.

The combination of these results severely limits the adjacency of universe elements. They are either adjacent to the root $r$, adjacent to some other universe element connected to a Steiner point, or are connected to some Steiner point whose corresponding set contains their corresponding element of $[n]$. In the former two cases, the universe element is a leaf node in $T$. In the latter case, while the universe element may not be a leaf node, if it is not, its only other neighbors are leaf node terminals. Treating those leaf node terminals as ``uncovered'' universe elements allows us to treat all terminals as leaf nodes.

The most important and most technical part of the proof is Lemma \ref{lem: graph st to st}. Lemma \ref{lem: graph st to st} shows that we may assume that there are no edges between Steiner points in the Steiner tree. In combination with the severely limited adjacencies of universe elements, this shows that we may assume that removing $r$ from the tree (and ignoring terminals not adjacent to Steiner points) divides the tree into groups of universe elements connected to Steiner points. By Proposition \ref{prop: graph term st edges}, these universe elements correspond to elements of $[n]$ contained in the sets corresponding to the Steiner points. So, the Steiner points adjacent to three terminals actually correspond precisely to a packing of sets from $\mathcal{S}$.

We want this induced packing to be large to yield insight into the size of maximum packing in the \setpT instance. Indeed, the induced packing corresponds to a maximum packing, as shown in Lemma \ref{lem: graph max packing}. Lemma \ref{lem: graph st deg at least 3} merely facilitates the proof of Lemma \ref{lem: graph max packing}. The rest of the proof follows from this correspondence.

 Now we proceed with the proof. Let $T = (T_V, T_E)$ be a minimum Steiner tree of $\tilde{P}$.  We wish to show that we may assume a nearly identical structure to the completeness case, particularly in that it is optimal to pack as many terminals into sets as possible. As we will show, these packings arise in the Steiner tree as groups of three terminals $t_i$, $t_j$, $t_k$ connected to a single Steiner point $s = \Gamma^{-1}(\{i,j,k\})$ which is in turn connected to $r$. In each successive claim, we assume that the properties of the previous claims initially hold.

 \begin{proposition} \label{prop: graph term st edges}
 We may assume that $(t_i, s) \in T_E$ for $t_i \in P$ and $s \in X$ only if $i \in \Gamma(s)$.
 \end{proposition}
 \begin{proof}
  Suppose $(t_i,s) \in T_E$ for $t_i \in P$, $s \in X$, and $i \not\in \Gamma(s)$. Then, $\Delta(t_i,s) = \beta_{\text{out}} \geq \trd, \srd$ from \ref{cond: sr prox} and \ref{cond: res s adj} of Definition \ref{defn: Steiner embeddability}. Remove $(t_i,s)$ from $T_E$. Then, one of the resulting connecting components contains $r$ so we may add edge $(r,s)$ or $(r, t_i)$ to reconnect the graph. The change in cost is then either $\srd - \beta_{\text{out}}$ or $\trd - \beta_{\text{out}}$ which, in either case, is at most $0$. This operation reduces the number of edges of the form $(t_i,s) \in T_E$ for $t_i \in P$, $s \in X$, and $i \not\in \Gamma(s)$, so repeating this process completes the proof since it yields a minimum cost $T$ satisfying the claim.
 \end{proof}

  \begin{lemma} \label{lem: graph deg condition term}
 We may assume that if $(t, t') \in T_E$  for some $t, t' \in P$, then $\deg(t) = 1$ or $\deg(t') = 1$ in $T$.
 \end{lemma}
 \begin{proof}
 Suppose otherwise. Let $(t,t') \in T_E$ with $t,t' \in P$ and $\deg(t), \deg(t') > 1$. Remove $(t,t')$. If the resultant connected component not containing $r$ contains a Steiner point $s$, we may add an edge $(r,s)$ to reconnect the tree. Since $\srd \leq \td$ by \ref{cond: sr prox}, the cost of the tree does not increase.
 
 Otherwise, without loss of generality assume that the connected component containing $t'$ does not contain a Steiner point and does not contain $r$. Hence it merely contains terminals $\{t_{i_1}, t_{i_2}, \ldots, t_{i_k}\} \subseteq P$. Remove all $(k-1)$ edges in this connected component. Each was of length $\td$. Now, add an edge $(t, t_{i_j})$ for $1 \leq j \leq l$. This reconnects the tree and fixes its original cost. 
 
 Note that in either case, the degree of each terminal is non-increasing. Hence, we reduced the number of edges between terminals with both terminals having degree $> 1$, so repeating this process yields the desired result. Also note that we do not add any edges of the form considered in Proposition \ref{prop: graph term st edges}, so the two claims are consistent.
 \end{proof}
 
 \begin{lemma} \label{lem: graph st existence}
 $T$ includes some $s \in X$. In particular, there exists $(s,t) \in T_E$ such that $s \in X$ and~$t \in P$.
 \end{lemma}
 \begin{proof}
 If there are no Steiner points in $T$, then we may assume that either every $t \in P$ is a leaf node connected to $r$ (in the case of $\trd < \td$) or there exists $\tilde{t} \in P$ such that $\tilde{t}$ is connected to $r$ and every other terminal is a leaf node with an edge to $\tilde{t}$. This is because, since $T$ is a tree, it has a fixed number of edges. Additionally, each edge is of length $\trd$ or $\td$. Then, in the former case, since $\trd < \td$, there is a unique tree with all edges of length $\trd$ (namely connecting every universe element to $r$), and the tree is minimum, that must be the tree. In the latter case, the minimum possible length of the tree is having one edge of length $\trd$ (to connect some universe element to $r$) and the rest of length $\td$. The described tree is such a tree. 
 
 Now, let $t_i, t_j, t_k \in P$ and $s \in X$ such that $\Gamma(s) = \{i,j,k\}$. We may assume w.l.o.g. that none of the terminals are $\tilde{t}$ in the latter case (using $m >3$). Then, remove the edges incident to $t_i, t_j, t_k$, disconnecting the graph into four connected components. We may reconnect the graph by introducing $s$ as a Steiner point, adding edges from $t_i$, $t_j$, and $t_k$ to $s$ and then adding an edge from $s$ to $r$. The change in cost is $3\beta_{\text{in}} + \srd - 3\min(\trd, \td)$. In either case the change in cost is negative, using \ref{cond: s util}, contradicting minimality of $T$.
 \end{proof}

  \begin{lemma} \label{lem: graph term leaf node}
 We may assume that $t \in P$ has $\deg(t) >1$ in $T$ only if $t$ is adjacent to a Steiner point.
 \end{lemma}
 \begin{proof}
  Suppose that $t \in P$ has degree greater than $1$ but is not adjacent to any Steiner points. By Lemma \ref{lem: graph deg condition term}, then all of the neighbors of $t$ are either $r$ or $t' \in P$ with $\deg(t') = 1$. 
  
  Now, if $\td > \trd$, removing all of $t$'s incident edges and connecting $t$ and all of its neighbors $t' \in P$ to $r$ reduces the cost of the tree, preserves connectivity, and decreases the degree of $t$ to $1$. Note that the number of edges removed equals the number of edges added. This is because $t$ must be adjacent to $r$ since all of its neighbors are leaf nodes from Lemma \ref{lem: graph deg condition term} and $T$ is connected.
  
  If $\td \leq \trd$, since there must be some $\tilde{t} \in P$ connected to a Steiner point $s$ by Lemma \ref{lem: graph st existence}, we can drop all of the edges incident to $t$ and add edges from $t$ and all its previously neighboring terminals to $\tilde{t}$. 
  This does not increase the cost of the tree since $\td \leq \trd$ in this case. It also preserves connectivity of the tree.
  
  Repeating this process then yields the claim. Note that this process is consistent with Proposition \ref{prop: graph term st edges} since it does not add any edges from Steiner points to terminals. This process is also consistent with Lemma \ref{lem: graph deg condition term} since one endpoint of every edge added between elements of $P$ in the second case is a leaf node.
 \end{proof}

 \begin{lemma} \label{lem: graph terminals only adj to one st}
 We may assume that each $t \in P$ is adjacent to at most one Steiner point.
 \end{lemma}
 \begin{proof}
  Suppose otherwise, let $t_i \in P$ be adjacent to more than one Steiner point. From Proposition \ref{prop: graph term st edges}, $t_i$ is adjacent only to $s \in S$ such that $i \in \Gamma(s)$. Drop edges from $t_i$ to Steiner points such that $t_i$ remains in the connected component with $r$. Then, add the same number of edges from $r$ to a Steiner point in each other connected component. This does not increase the cost of $T$ since $\srd \leq \beta_{\text{in}}$ from \ref{cond: sr prox}. Additionally, this process reduces the number of elements of $P$ adjacent to more than one Steiner point, so repeating this completes the claim.
  
  This process is consistent with Proposition \ref{prop: graph term st edges} since none of the edges considered in that claim are added. It is consistent with Lemmas \ref{lem: graph deg condition term} and \ref{lem: graph term leaf node} since the degree of non root terminals only decreases in this process.
 \end{proof}

 \begin{proposition} \label{prop: graph terminals not adj to r if in triangle}
 We may assume that if $t \in P$ such that $t$ is adjacent to a Steiner point $s$, then $t$ is not adjacent to $r$.
 \end{proposition}
 \begin{proof}
 Suppose that $(t,r) \in T_E$. Drop $(t,r)$ and add edge $(s,r)$. This reconnects the tree and does not increase its cost since $\srd \leq \trd$ by \ref{cond: sr prox}. As usual, repeating this process yields the desired claim. Moreover, this process is consistent with the above claims. We do not add an edges of the form to violate Proposition \ref{prop: graph term st edges} or Lemma \ref{lem: graph terminals only adj to one st}, and we only decrease the degree of terminals so this process is consistent with Lemmas \ref{lem: graph deg condition term} and \ref{lem: graph term leaf node}.
 \end{proof}

 \begin{lemma} \label{lem: graph st to st}
 We may assume that $(s, s') \not\in T_E$ for $s, s' \in X$.
 \end{lemma}
 \begin{proof}
  First, if $\gamma_0, \gamma_1, \gamma_2 \geq \srd$, then the result follows easily. Suppose $(s,s') \in T_E$ for $s, s' \in X$. Then, remove the edge between them. Assume w.l.o.g. that $s$ is in the resultant connected component with $r$. Then, add edge $(s', r)$ to reconnect the graph. Since $\gamma_i = \Delta(s,s') \geq \Delta(s', r) = \srd$, the cost of the tree does not increase under this operation.
  
  If instead $\gamma_0, \gamma_1 \geq \srd$ but $\gamma_2 < \srd$ (the only other case by \ref{cond: sr prox}), we proceed carefully. By applying the procedure above we may assume that $(s,s') \in T_E$ only if $|\Gamma(s) \cap \Gamma(s')| = 2$. Fix $s \in X$ with edges to other Steiner points. We may assume that $s$ has at most $3$ neighboring Steiner points. This is because each of $s$'s neighbors overlap $\Gamma(s) = \{i,j,k\}$ on two elements by Proposition \ref{prop: graph term st edges}, and there are three choices of pairs of elements to overlap on ($\{i,j\}$, $\{i,k\}$, or $\{j,k\}$). Hence, for all but one neighbor per pair of overlapping set elements, the edges to $s$ can be adjusted to be edges to one another. For example, if $s'$ and $s''$ are Steiner points adjacent to $s$ such that $\{i,j\} \subset \Gamma(s'), \Gamma(s'')$, then $\Delta(s', s'') = \Delta(s, s') = \Delta(s, s'') = \beta_{\text{in}}$ so the edge $(s, s'')$ may be replaced by the edge $(s', s'')$ without changing the cost or connectivity of the tree.
  
  Similarly, we may assume that $s$ is not adjacent to $r$ by moving an edge from $s$ to $r$ to be an edge from $r$ to one of $s$'s neighboring Steiner points. Now we consider two cases based on the number of neighboring Steiner points to $s$.

  \begin{paragraph}{Case 1: $s$ is adjacent to two or three Steiner points.}
  If $s$ is adjacent to at least two Steiner points $s'$ and $s''$, then $\Gamma(s) \subseteq \Gamma(s') \cup \Gamma(s'')$ since $|\Gamma(s) \cap \Gamma(s')| = 2$, $|\Gamma(s) \cap \Gamma(s'')| = 2$, and $\Gamma(s),\Gamma(s'),\Gamma(s'')$ are distinct sets of size $3$. Hence, using Proposition \ref{prop: graph term st edges}, for each non-root terminal adjacent to $s$, we may remove its edge to $s$ and add a new edge to $s'$ or $s''$ (depending on which induces less cost) without changing the connectivity or cost of $T$.

  At this point, $s$ has degree $2$ or $3$. If $s$ has degree $2$, $s$'s only neighbors are two Steiner points $s'$ and $s''$ such that $|\Gamma(s') \cap \Gamma(s'')| = 1$. Drop $s$ and its incident edges. Add an edge from one of $s$'s former neighbors to $r$ to reconnect $T$. This induces an additive cost of $\srd - 2\gamma_2 \leq 0$. This inequality holds since $2\gamma_2 \geq \gamma_1$ by the triangle inequality and $\srd \leq \gamma_1$ by \ref{cond: sr prox}.
  
  So, we may assume that $s$ has degree exactly $3$. But consider removing $s$ and its incident edges. This reduces the cost of $T$ by $3\gamma_2$. To reconnect $T$, we add edges from $r$ to a Steiner point in each of the connected components of $T$ not containing $r$. Both components may be clearly contain Steiner points. This costs $2\srd$. Hence, since $3\gamma_2/2 \geq \srd$ by \ref{cond: sr prox}, this does not increase the cost of the tree.
  \end{paragraph}
  
  \begin{paragraph}{Case 2: $s$ is adjacent to exactly one Steiner point.}
  Suppose $s$ is only adjacent to one other Steiner point, $s'$. Then, since $|\Gamma(s) \cap \Gamma(s')| = 2$, we can remove edges from $s$ to all but one of its neighboring universe elements and replace those with edges from the terminals to $s'$ without changing the connectivity or cost of $T$. Now, we have that $s$ is not adjacent to $r$, $s$ is adjacent to at most one Steiner point (our present case), and $s$ is adjacent to at most one universe element. Then, remove $s$ and its incident edges from $T$. If $s$ was not adjacent to any terminals, the tree remains connected and its cost decreases (and note that $s$ is not adjacent to $r$).
  
  Otherwise, suppose that $s$ was adjacent to a universe element $t$. If $t$ had degree greater than $1$, then, by Lemmas \ref{lem: graph term leaf node} and \ref{lem: graph terminals only adj to one st} and Proposition \ref{prop: graph terminals not adj to r if in triangle}, $t$'s other neighbors are exactly leaf node universe elements. Remove the edges from those terminals to $t$ and replace those edges with edges to some other terminal adjacent to a Steiner point (for example, some terminal adjacent to $s'$). This fixes the cost of the tree.
  
   Now, add an edge from $t$ to either a terminal adjacent to a Steiner point or $r$ (whichever edge is cheaper). This changes the cost of the tree by an additive factor of
  \[
  \min(\trd, \td) - \gamma_2 - \beta_{\text{in}} \leq 0
  \]
  and reconnects the tree. The inequality holds by \ref{cond: ss' non adj}.
  \end{paragraph}

  Repeating this process removes all edges between Steiner points since it reduces the number of edges between Steiner points by at least one in each iteration. 
  
  Now we check that this process is consistent with the above claims. We do not add any edges to violate Proposition \ref{prop: graph term st edges}.  This process is consistent with Lemmas \ref{lem: graph deg condition term} and \ref{lem: graph term leaf node}  since the only time we add an edge between non-$r$ terminals in this process is in Case 2 and all such added edges are between a leaf node and a terminal adjacent to a Steiner point. This process is consistent with Lemma \ref{lem: graph terminals only adj to one st} since we only replace edges from non-$r$ terminals to Steiner points with different Steiner point endpoints than adding additional edges to Steiner points. Finally, the only time we add an edge from a non-$r$ terminal to $r$ is in Case $2$, and we only do this when the non-$r$ terminal is a leaf node. Hence, this process is consistent with Proposition \ref{prop: graph terminals not adj to r if in triangle}.
 \end{proof}
 
One important consequence of the above claims is that any path between two Steiner points in $T$ must pass through $r$. That facilitates computing the cost of the tree. Before completing the proof, we prove a lemma that facilitates our analysis. 

\begin{lemma} \label{lem: graph st deg at least 3}
We may assume that each Steiner point has degree at least $3$ in $T$.
\end{lemma}
\begin{proof}
 If there is a Steiner point of degree $1$, dropping the Steiner point and its incident edge yields a lower cost tree, violating minimality of $T$. Now suppose that $s$ is a Steiner point of degree $2$ in $T$. By Lemma \ref{lem: graph st to st} and Proposition \ref{prop: graph term st edges}, $s$'s only possible neighbors are $r$ and $t_i \in P$ such that $i \in \Gamma(t_i)$. One of the two neighbors must be $r$ since $T$ is connected and, by Lemma \ref{lem: graph term leaf node} and Proposition \ref{prop: graph terminals not adj to r if in triangle}, none of $s$'s neighboring universe elements nor their neighboring universe elements may be adjacent to $r$. Then, $s$'s neighbors are precisely $r$ and some universe element. 

If $\trd \leq \td$, then we may assume that $s$'s neighboring universe element is a leaf node. In that case, drop $s$ and its incident edges and connect $s$'s neighboring universe element to $r$. This does not increase the cost of the tree by the triangle inequality and maintains connectivity.

If $\td < \trd$, let $t$ be $s$'s universe element neighbor. Remove $s$ and its incident edges and add an edge from $r$ to $t$. As above, this retains connectivity and does not increase the cost of the tree by the triangle inequality. Now, the tree must have an additional Steiner point by Lemma \ref{lem: graph st existence} (otherwise we contradict minimality). That Steiner point must be adjacent to a universe element $\tilde{t}$ since otherwise it is degree $1$. But replacing the edge $(t,r)$ with an edge $(t, \tilde{t})$ then reduces the cost of the tree while preserving connectivity, contradicting minimality of the tree.

Hence, degree $2$ Steiner points are only even possible when $\trd \leq \td$. In that case, this simple modification is evidently consistent with all of the above claims.
\end{proof}
 
 \begin{lemma} \label{lem: graph max packing}
The Steiner points in $T$ adjacent to $3$ terminals correspond exactly to a maximum packing of sets in $\mathcal{S}$.
 \end{lemma}
 \begin{proof}
By Proposition \ref{prop: graph term st edges} and Lemma \ref{lem: graph terminals only adj to one st}, we know that a given terminal $t_i$ is adjacent to $s$ only if $i \in \Gamma(s)$ and is adjacent to at most one such $s$. Then, since every $S \in \mathcal{S}$ has $|S| = 3$, the Steiner points adjacent to $3$ terminals correspond exactly to a disjoint choice of sets from $\mathcal{S}$.

It remains to prove that this choice of sets is of maximum size. To show this, we need to understand the cost of $T$. 

Partition the edges of $T$ into three sets $E_1, E_2,$ and $E_3$ defined as follows: edges incident to universe elements not adjacent to Steiner points, edges incident to Steiner points adjacent to two universe elements, and edges incident to Steiner points adjacent to three universe elements. 
Lemma \ref{lem: graph deg condition term} and Proposition \ref{prop: graph terminals not adj to r if in triangle} imply that every edge without a leaf universe element involves a Steiner point. Then, Proposition \ref{prop: graph term st edges}, Lemma \ref{lem: graph st deg at least 3}, and the fact that the sets in $\mathcal{S}$ are each cardinality $3$ shows that this accounts for all for all edges in $T$.

Now, let $x_1$ be the number of universe elements not adjacent to Steiner points, $x_2$ the number of universe elements adjacent to Steiner points adjacent to two universe elements, and $x_3$ the number of universe elements adjacent to Steiner points adjacent to three universe elements. By Lemma \ref{lem: graph deg condition term} and Proposition \ref{prop: graph terminals not adj to r if in triangle} combined with the restrictions on the adjacency of Steiner points from Proposition \ref{prop: graph term st edges} and Lemma \ref{lem: graph st to st}, every Steiner point must be adjacent to $r$ since $T$ is connected. Then, the total cost contributed by edges in $E_1$ is $\min(\trd, \td)x_1$, $(\srd/2 + \beta_{\text{in}})x_2$ for $E_2$, and $(\srd/3 + \beta_{\text{in}})x_3$ for $E_3$.

Then, the cost of $T$ is precisely 
\[
\min(\trd, \td) x_1  + (\beta_{\text{in}} + \srd/2)x_2 + (\beta_{\text{in}} + \srd/3)x_3.
\]
But, since $\beta_{\text{in}} + \srd/3 < \min(\trd, \td)$
by \ref{cond: s util} and $\beta_{\text{in}} + \srd/3 <\beta_{\text{in}} + \srd/2$,
this cost is minimized for maximizing $x_3$. This means that the set of Steiner points in $T$ adjacent to $3$ terminals must correspond to a maximum packing of sets in $\mathcal{S}$ (by minimality of the cost of $T$).
 \end{proof}
 
 As a result of Lemma \ref{lem: graph max packing}, we have the following.
 \begin{corollary}
 In the soundness case we have
\[
\cost_\Delta(T) \geq n(\srd/3+\beta_{\text{in}})\cdot \left(1 - \varepsilon  + \frac{(4\varepsilon  - 2\delta)}{3} \cdot \frac{\min(\trd, \td, \beta_{\text{in}} + \srd/2)}{\srd/3+\beta_{\text{in}}} + \frac{(2\delta  - \varepsilon)}{3} \cdot \frac{\min(\trd, \td)}{\srd/3+\beta_{\text{in}}} \right).
\]
 \end{corollary}
 \begin{proof}
 Define $x_1, x_2,$ and $x_3$ as in the proof of Lemma \ref{lem: graph max packing} and let $P_1, P_2,$ and $P_3$ partition $P$ into the universe elements counted by $x_1$, $x_2$, and $x_3$.  From the proof of Lemma \ref{lem: graph max packing}, we have that the cost of the $T$ is given by 
 \[
 \cost_{\Delta}(T) = \min(\trd, \td) x_1  + (\beta_{\text{in}} + \srd/2)x_2 + (\beta_{\text{in}} + \srd/3)x_3.
 \]
 From Lemma \ref{lem: graph max packing}, this is minimized for maximum $x_3$. Now, as stated in Lemma \ref{lem: graph max packing}, the sets in $\mathcal{S}$ corresponding to each Steiner point adjacent to three universe elements correspond to a packing of sets. Since we are in the soundness case,  we then have $x_3 \leq (1 - \varepsilon)n$. Since we are lower bounding the cost of $T$ and maximizing $x_3$ minimizes the cost of $T$, suppose  $x_3 = (1 - \varepsilon)n$. 
 
 Now we make use of $\delta$, describing how to construct a set cover $\mathcal{S}'$ from $T$. For each Steiner point $s \in T$, add $\Delta(s)$ to $\mathcal{S}'$. For each remaining uncovered universe element, add some set covering that element. This yields a set cover of size at most $x_1 + x_2/2 + x_3/3$. Since we are in the soundness case (and using $x_3 = (1-\epsilon)n$) we have that
 \[
 x_1 + x_2/2 + (1-\epsilon)(n/3) \geq (1+\delta)(n/3)
 \]
 and 
\[
x_1 + x_2 + (1-\epsilon)n = n.
\]
Then, this implies $(4\varepsilon  - 2\delta)/3 \cdot n \geq x_2$. Our bound on $\cost_{\Delta}(T)$ then follows by maximizing $x_2$ when $\beta_{\text{in}} + \srd/2 < \min(\trd, \td)$.
 \end{proof}
 \end{paragraph}

\section{\texorpdfstring{\apx}{APX}-hardness of  \texorpdfstring{\dst}{DST} in \texorpdfstring{$\ell_p$}{lp}-metrics} \label{sec: l_p}

In this section, we describe the embedding of the metric Steiner tree instances (described in the previous section), into $\R^n$ equipped with the $\ell_p$-metric for $p \in (1, \infty]$. The \apx-hardness when $p = 1$ follows easily from the known \apx-hardness of  \cst in $\ell_1$-metric or $\ell_0$-metric given in  \cite{Trevisan00} (or see Theorem~\ref{thm: vx cover corr} for a simplified reduction). 

\subsection{An Embedding into \texorpdfstring{$\ell_p$}{lp}-metric Spaces} \label{subsec: ell_p embedding}

For a given $p$, our corresponding instance of  \dst in $\ell_p$-metric will be as follows. For an instance  $([n],\mathcal{S})$ of $(\varepsilon,\delta)$-\setpT where $|\mathcal{S}|=m$, let $P = \{\mathbf{e_i} \,:\, i \in [n] \}$. Then, the set of terminals are defined to be $\tilde{P} = P \cup \{\mathbf{0}\}$ (with $r = \mathbf{0}$ the distinguished ``root'' terminal). Let $\pcoord$ be a scalar defined for each $p$. Then,  the set of facilities is defined to be 
\[
X = \{ \pcoord \cdot (\mathbf{e_i} + \mathbf{e_j} + \mathbf{e_k}) \,:\, \{i,j,k\} \in \mathcal{S} \}.
\]
We choose $\pcoord$ such that the following properties hold:
\begin{enumerate}
    \item $0 < \pcoord \leq 1/2$ and 
    \item $3 ((1-\pcoord)^p + 2\pcoord^p)^{1/p} + (3\pcoord^p)^{1/p} < 3.$
\end{enumerate}
Note that this is possible since, for
\[
f_p(x) = 3 ((1-x)^p + 2x^p)^{1/p} + (3x^p)^{1/p} - 3
\]
we have $f_p(0) = 0$ and $\frac{df_p(0)}{dx} = -3 + 3^{1/p} < 0$ for $p > 1$. This point configuration induces a complete graph $G = (\tilde{P} \cup X, E)$ with edge costs given by distances in the $\ell_p$-metric. 

Now, we check that the conditions of Theorem \ref{thm: graph hardness} hold. Clearly our choices of distances are metric compatible since they arise from a metric space. 
\begin{remark} \label{rmk: l_p distances}
Note that, for $p \in (1, \infty)$,
\begin{multicols}{2}
\begin{itemize}
    \item $\srd= 3^{1/p}\pcoord$,
    \item $\trd= 1$,
    \item $\beta_{\text{in}} = ((1-\pcoord)^p + 2\pcoord^p)^{\frac{1}{p}}$,
    \item $\beta_{\text{out}} = (1 + 3\pcoord^p)^{\frac{1}{p}}$,
    \item $\gamma_0 = 6^{1/p}\pcoord$,
    \item $\gamma_1 = 4^{1/p}\pcoord$,
    \item $\gamma_2 = 2^{1/p} \pcoord$, and 
    \item $\td = 2^{1/p}$.
\end{itemize}
\end{multicols}

When $p = \infty$, we instead have
\begin{itemize}
    \item $\srd = \gamma_0 = \gamma_1 = \gamma_2 = \pcoord$,
    \item $\trd = \beta_{\text{out}} = \td =  1$, and
    \item $\beta_{\text{in}} = 1-\pcoord$.
\end{itemize}
\end{remark}

\begin{proposition} \label{prop: l_p Steiner embed}
For each $p \in (1, \infty]$, the corresponding tuple $(\srd, \trd, \beta_{\text{in}}, \beta_{\text{out}}, \gamma_0, \gamma_1, \gamma_2, \td)$ from the above is Steiner embeddable.
\end{proposition}
\begin{proof}
For all $p$, the tuple is clearly metric compatible since it arises from a metric space. It clear that \ref{cond: sr prox}, \ref{cond: res s adj}, \ref{cond: s util}, and \ref{cond: ss' non adj} hold when $p = \infty$. Now we prove that these constraints hold when $p$ is finite. Note first 
that
\begin{align*}
3^{1/p}\pcoord &\leq ((1-\pcoord)^p + 2\pcoord^p)^{1/p} \\
&< ((1-\pcoord)^p + 2\pcoord^p)^{1/p} + 3^{1/p}\pcoord/3 < 1 \leq 2^{1/p}.
\end{align*}
The first inequality follows from $0 < \pcoord \leq 1/2$ and the last inequality follows from $3 ((1-\pcoord)^p + 2\pcoord^p)^{1/p} + (3\pcoord^p)^{1/p} < 3.$ This immediately implies $\srd \leq \trd, \beta_{\text{in}}, \td$. Since $0 < \pcoord \leq 1/2$, $\beta_{\text{in}} < \beta_{\text{out}}$, and $\srd \leq \gamma_1, \gamma_2$ is clear. Also observe that $_{st} \leq 3\gamma_2/2$ if and only if $(3/2)^{\frac{1}{p}} \leq 3/2$ which holds for all $p \geq 1$. This proves \ref{cond: sr prox} holds. The inequality $\trd \leq \beta_{\text{out}}$ is immediate (showing \ref{cond: res s adj}) and $\beta_{\text{in}} + \frac{\srd}{3} < \min(\trd, \td)$ holds by our choice of $\pcoord$ (showing \ref{cond: s util}). 

It then remains to show $\min(\trd, \td) \leq \beta_{\text{in}} + \gamma_2$, \ref{cond: ss' non adj}. But note that 
\[
1 = (1 - \pcoord) + \pcoord \leq \beta_{\text{in}} + \gamma_2,
\]
completing the proof.
\end{proof}

Then, applying Theorem \ref{thm: graph hardness} to Proposition \ref{prop: l_p Steiner embed}, we have the following.
\begin{theorem} \label{thm: hardness of lp dst}
Let $p \in (1, \infty)$. Suppose that 
\begin{enumerate}
    \item $0 < \pcoord \leq 1/2$,
    \item $3[(1-\pcoord)^p + 2\pcoord^p]^{\frac{1}{p}} + 3^{\frac{1}{p}}\pcoord < 3$,
    \item $(\varepsilon, \delta)$-\setpT is \np-hard.
\end{enumerate}
Then,  \dst in the  $\ell_p$-metric  is \np-hard to approximate within a factor of 
\[
(1 - \varepsilon) + \frac{(4\varepsilon  - 2\delta)}{3} \cdot \frac{\min(1, [(1-\pcoord)^p + 2\pcoord^p]^{\frac{1}{p}} +\frac{3^{1/p}}{2}\pcoord)}{[(1-\pcoord)^p + 2\pcoord^p]^{\frac{1}{p}} + 3^{\frac{1-p}{p}}\pcoord} + \frac{(2\delta  - \varepsilon)}{3} \cdot \frac {1}{[(1-\pcoord)^p + 2\pcoord^p]^{\frac{1}{p}} + 3^{\frac{1-p}{p}}\pcoord}> 1.
\]
In particular, for all $p \in (1,\infty)$, there exists $\pcoord$ such that the first two conditions above hold and Theorem~\ref{thm: triangle cov hardness} implies $(\varepsilon_0,\varepsilon_0/2)$-\setpT is \np-hard for some $\varepsilon_0>0$.
\end{theorem}
\begin{remark} \label{rmk: role of eps and delta}
Curiously, we observe that the importance of $\varepsilon$ and $\delta$ depends dramatically on $p$. For $p \in (1, 1/\log_3(4/3))$, $\delta$ is sometimes irrelevant. This occurs when $\theta_p$ is chosen such that $1 < [(1-\pcoord)^p + 2\pcoord^p]^{\frac{1}{p}}$. Intuitively, this means that in the soundness case, the optimal Steiner tree not only corresponds to a maximal set packing (as proved in Lemma~\ref{lem: graph max packing}), but also that the terminals (universe elements) that are not yet covered by the packing, would like to connect by an edge directly with the root in the optimal Steiner tree, instead of having two terminals (universe elements) connecting to a common Steiner point (which is then connected to the root). This happens for $\theta_2 = 1/6$, for example. See Figure \ref{fig: eps delta depend}.
\end{remark}

By applying the Fr\'echet embedding to the hard instance in Theorem \ref{thm: graph hardness},  we have that only $\delta$ is relevant to the hardness of  \dst in $\ell_{\infty}$-metric (and not $\varepsilon$). 
In many other settings, it appears that both $\varepsilon$ and $\delta$ play a role with larger $\delta$ permitting better hardness of approximation. Below we prove the non-dependency of $\varepsilon$ for \dst in $\ell_\infty$-metric, without invoking the Fr\'echet embedding.

\begin{theorem} \label{thm: graph l_inf hardness}
If $(\varepsilon, \delta)$-\setpT is \np-hard, then \dst in $\ell_{\infty}$-metric is \np-hard to approximate within a factor of $1 + \delta/4$.
\end{theorem}
\begin{proof}
Note that $\theta_{\infty} = 1/2$ satisfies the required properties for Proposition \ref{prop: l_p Steiner embed}. Then, from Corollary \ref{cor: general metric useless?}, it suffices to simply check that, for $p = \infty$, $\beta_{\text{in}} + \srd/2 = \min(\trd, \td, \beta_{\text{in}} + \srd/2)$, $\min(\trd, \td) = \srd + \beta_{\text{in}}$, and $\beta_{\text{in}} = \srd$. This is clear from observing the distances in Remark \ref{rmk: l_p distances}.
\end{proof}

\begin{corollary}
For all $p \in [1, \infty]$,  \dst in $\ell_p$-metric is \apx-hard.
\end{corollary}
\begin{proof}
The case of \dst in $\ell_1$-metric follows from the known \apx-hardness of  \cst in $\ell_1$-metric or $\ell_0$-metric given in  \cite{Trevisan00} (or see Theorem~\ref{thm: vx cover corr}).   The cases of $p \in (1, \infty)$ and $p = \infty$ follow from Theorems \ref{thm: hardness of lp dst} and \ref{thm: graph l_inf hardness}, respectively.
\end{proof}

\subsection{Hardness of Approximating  \texorpdfstring{\dst}{DST} in \texorpdfstring{$\ell_p$}{lp}-metrics}\label{sec:figureexp}
Optimizing for the choice of $\pcoord$ in the construction from Proposition \ref{prop: l_p Steiner embed}, we have the following.
\begin{corollary} \label{cor: l_p tight hardness}
If $(\varepsilon, \delta)$-\setpT is \np-hard, then for $p \in [\frac{1}{\log_3(4/3)}, \infty]$,  \dst in $\ell_p$-metric is \np-hard to approximate within a
factor
\[ 
1+ \varepsilon \left(\frac{1}{2}-\frac{1}{2\cdot3^{1/p}} \right) + 2\delta \left(\frac{1}{2\cdot3^{1/p}}-\frac{3}{8} \right).
\]
In particular, we have  \dst in $\ell_p$-metric is \np-hard to approximate within a factor of $1 + \varepsilon/8$.
\end{corollary}

\begin{proof}
The second part of the claim occurs when $\delta = \varepsilon/2$. In this case, the second fractional term in Theorem \ref{thm: hardness of lp dst} is nullfied. Then, observe that, for a fixed $p$, while 
\[
[(1-\pcoord)^p + 2\pcoord^p]^{\frac{1}{p}} +\frac{3^{1/p}}{2}\pcoord \leq 1,
\]
the hardness of approximation factor is maximized for maximal $\pcoord$. By our constraints on $\pcoord$, $\pcoord \leq 1/2$. When $\pcoord = 1/2$ is valid we have
\[
3[(1-1/2)^p + 2(1/2)^p]^{\frac{1}{p}} + 3^{\frac{1}{p}}(1/2) = 2\cdot3^{1/p} < 3.
\]
This then holds for  $\frac{1}{\log_3(3/2)} < p.$
Additionally, 
\[
[(1-\pcoord)^p + 2\pcoord^p]^{\frac{1}{p}} +\frac{3^{1/p}}{2}\pcoord \leq 1,
\]
holds for $\pcoord = 1/2$ if and only if  $3^{1/p} \cdot \frac{3}{4} \leq 1$ so $p \geq \frac{1}{\log_3(4/3)}.$ Hence, for $p \in [\frac{1}{\log_3(4/3)}, \infty]$, the embedding with $\pcoord = 1/2$ yields that   \dst in $\ell_p$-metric is hard to approximate within a factor of $1 + \varepsilon/8$.

If $(\varepsilon, \delta)$-\setpT is \np-hard for $\delta > \varepsilon/2$, we get improved hardness for $p \in [\frac{1}{\log_3(4/3)}, \infty]$ by using $\theta_p = 1/2$. Computing yields the claimed expression.
\end{proof}

When $\delta > \varepsilon/2$, optimizing the hardness exactly is more challenging. This is because, for $p$ finite, the maximum value of 
\[
\frac{1}{[(1-\pcoord)^p + 2\pcoord^p]^{\frac{1}{p}}}
\]
arises for $\theta_p < 1/2$, so the optimal value depends intimately on $\delta$. Likewise, when $p \in (1, \frac{1}{\log_3(4/3)})$, we need $\pcoord < 1/2$. At that point even finding a permissible choice of $\pcoord$ requires an explicit computation. 

We summarize the dependence of the hardness of  \dst in $\ell_p$-metric  on $\varepsilon$ and $\delta$ (and therefore the dependence on the hardness of Set Packing and Set Cover) in Figure \ref{fig: eps delta depend}. 
The purple region, $p \in \{1, \infty\}$, is the range where there is no dependency on $\varepsilon$, and hence the reduction only uses the hardness of set packing (see Corollary \ref{cor: general metric useless?} and Theorem \ref{thm: vx cover corr}). In the red range, $[1/(\log_3(4/3), \infty)$, we show in Corollary \ref{cor: l_p tight hardness} a dependency on both $\varepsilon$ and $\delta$.

The behavior in the yellow range is less clear. For $p$ such that the optimal choice of $\theta_p$ satisfies
\[
[(1-\pcoord)^p + 2\pcoord^p]^{\frac{1}{p}} +\frac{3^{1/p}}{2}\pcoord < 1,
\]
the hardness depends both on $\delta$ and $\varepsilon$; larger $\delta$ will yield increased hardness of approximation. This is, however, not the case for all values of $p$ in the yellow region. The choice of $p=2$ is a notable exception (see Corollary \ref{cor: Euc hardness}). In the green region there is no dependence on $\delta$. This is because the optimal choice of $\theta_p$ satisfies
\[
[(1-\pcoord)^p + 2\pcoord^p]^{\frac{1}{p}} +\frac{3^{1/p}}{2}\pcoord \geq 1.
\]

~\section{Reduction from \texorpdfstring{\cst}{CST} to \texorpdfstring{\dst}{DST}} \label{sec: reduction from cst to dst}
In this section, we link \cst to \dst, proving that \dst is computationally harder than \cst. The sense in which this holds is formalized in Theorem \ref{thm: reduction from cst to dst}.
\begin{theorem} \label{thm: reduction from cst to dst}
Let $(\mathcal{X}, d)$ be a metric space such that there exists an algorithm for \cst in $(\mathcal{X}, d)$ that runs in $O(f(n))$ time for some computable function $f$ and $n$ the number of input terminals. Let $P$ be an instance of \cst in $(\mathcal{X}, d)$, and let $T^*$ be an optimal Steiner tree of $P$. Then, for any $\varepsilon > 0$, there exists a $\poly(n, \varepsilon)$-time algorithm outputting $X \subset \mathcal{X}$  with the following properties.
\begin{enumerate}
    \item $|X| = \poly(n, \varepsilon).$
    \item If $T$ is the optimal Steiner tree on \dst instance $(P, X)$, then 
    \[
    \cost(T^*) \leq \cost(T) \leq (1 + \varepsilon)\cost(T^*).
    \]
\end{enumerate}
\end{theorem}

To prove this, we appeal to a powerful structural result about optimal Steiner trees.
\begin{theorem}[\cite{DuZhangFeng91, Borchers97, Bartal_Gottlieb_2021}] \label{thm: bg subtree thm}
Let $P$ be an instance of \cst in metric space $(\mathcal{X}, d)$ and let $T^*$ be an optimal Steiner tree of $P$. For any $\varepsilon > 0$, there exists constant $C_{\varepsilon}$ and Steiner tree $T'$ of $P$ with the following properties.
\begin{enumerate}
    \item $\cost(T^*) \leq \cost(T') \leq (1 + \varepsilon)\cost(T^*)$.
    \item $T'$ is formed by  partitioning $P$ into parts of size at most $C_{\varepsilon}$, finding optimal Steiner trees of each part, and connecting those Steiner trees via edges between their respective terminals.
\end{enumerate}
\end{theorem}
This follows from Theorem 3.1 of \cite{DuZhangFeng91} with $C_{\varepsilon} \leq 2^{1 + \varepsilon^{-1}}$. The best possible $C_\varepsilon$ can be derived from Theorem 3.2 of \cite{Borchers97}. This result can also be derived from the proof of Theorem 3.2 of \cite{Bartal_Gottlieb_2021} (with somewhat weaker $C_{\varepsilon} \leq 2^{32/\varepsilon \ln (8/\varepsilon)}$).

\begin{proof}[Proof of Theorem \ref{thm: reduction from cst to dst}.]
Consider every subset of $P$ of size at most $C_{\varepsilon/2}$ as an instance of \cst in $(\mathcal{X}, d)$. There are $\poly(n, \varepsilon)$ such instances and optimal Steiner trees (optimal up to an arbitrarily small factor $1 + \varepsilon'$ due to numerical precision) can be computed in constant time (in $\varepsilon$ and $\varepsilon'$). Let $X$ be the union of the Steiner points used in the Steiner tree solutions for these instances. Note that $|X| = \poly(n)$ since Steiner trees on $n$ terminals may be assumed to use at most $n-2$ Steiner points (this follows by the triangle inequality and a degree counting argument—see Section 3.4 of \cite{Gilbert_Pollak_1968}, for example). Let $T'$ be the near-optimal Steiner tree from  Theorem \ref{thm: bg subtree thm}. Let $Q$ be some set of terminals in the partition of terminals induced by $T'$. Then, the optimal Steiner tree of \dst instance $(Q, X)$ costs at most a $(1 + \varepsilon')$ factor more than the cost of the optimal Steiner tree for \cst instance $Q$. Then, since this holds for each part in the partition of terminals induced by $T'$ from Theorem \ref{thm: bg subtree thm}, if $T$ is the optimal Steiner tree on \dst instance $(P, X)$, then $\cost(T) \leq (1 + \varepsilon') \cost(T') \leq (1 + \varepsilon') (1 + \varepsilon/2)  \cost(T^*)$. Since $\varepsilon'$ can be made arbitrarily small, this implies the desired result.
\end{proof}

In particular, Theorem \ref{thm: reduction from cst to dst} holds for all $\ell_p$-metric spaces and in any dimension. This shows that \cst is essentially no harder than \dst in those spaces.
\begin{corollary} \label{cor: reduction from CST to DST for l_p metrics}
For all $\varepsilon > 0$ and $p \in [1, \infty]$, if there exists a polynomial time $(1 + \alpha)$-approximation algorithm for \dst in the $\ell_p$-metric, then there exists a polynomial time $(1 + \alpha + \varepsilon)$-approximation algorithm for \cst in the $\ell_p$-metric.
\end{corollary}
\begin{proof}
From Theorem \ref{thm: reduction from cst to dst}, it suffices to show that \cst in the $\ell_p$-metric admits an algorithm that runs in $O(f(n))$ time for some computable function $f$ and $n$ the number of input terminals. Given a set of terminals $P$ with $|P| = n$, consider the subspace of dimension at most $n$ formed by those points. Then, the approximation scheme of Arora (see 1.1.1 of \cite{Arora_1998}) runs in $O(f(n))$ time for a computable function $f$, yielding the result.
\end{proof}

Theorem \ref{thm: reduction from cst to dst} and Corollary \ref{cor: reduction from CST to DST for l_p metrics} give a construction for efficient approximation algorithms of \cst using the literature on \dst.
\begin{corollary}\label{cor:alg}
For all $\varepsilon > 0$ and $p \in [1, \infty]$,  \cst in the $\ell_p$-metric admits a polynomial time $(\ln(4) + \varepsilon)-$approximation algorithm.
\end{corollary}
\begin{proof}
Run the algorithm implicit in \ref{thm: reduction from cst to dst} and Corollary \ref{cor: reduction from CST to DST for l_p metrics} to compute a \dst instance. Then, apply the $(\ln(4) + \varepsilon')-$approximation algorithm for \dst from \cite{byrka2017improved} and choose $\varepsilon'$ to be appropriately small. 
\end{proof}

 \section*{Acknowledgement}
We would like to thank Fabrizio Grandoni for discussion about \dst.
\newpage

\printbibliography
\newpage

\appendix
\section{Explicit Gaps for \texorpdfstring{\dst}{DST} in Specific \texorpdfstring{$\ell_p$}{lp}-metrics}\label{sec:explicit}
In this section, we compute explicit hardness of approximation constants by plugging in the best inapproximability results known for \setpT and \vcov in literature.
\subsection{Hardness of Approximation Bounds for Euclidean \texorpdfstring{\dst}{DST}}
As the Euclidean metric is of particular importance, we explicitly state our resultant hardness factor for $p = 2$. In this special case, we can make do with a somewhat weaker completeness case for \setpT than that in $(\varepsilon, \delta)$-\setpT. Instead we consider the decision problem $(a,b)$-Gap \setpT 
\begin{definition}[$(a,b)$-Gap \setpT]
Given a set system $([n], \mathcal{S})$ where  for all $S \in \mathcal{S}$, we have $|S| = 3$, and $|\mathcal{S}| = m$, the  $(a,b)$-Gap \setpT is  the problem of deciding which of the following cases hold.
\begin{itemize}
    \item \textit{Completeness:} There exists a subcollection of $\mathcal{S}$ which are pairwise disjoint and cover at least $(1-a)n$ elements in $[n]$.
    \item \textit{Soundness:} Any subcollection of pairwise disjoint sets of $\mathcal{S}$ cover at most $(1-b)n$ elements in $[n]$.
\end{itemize}
\end{definition}
Observe that $(\varepsilon, \varepsilon/2)$-\setpT is actually $(1, 1-\varepsilon)$-Gap \setpT since $\delta$ is unnecessary in this formulation. Then, we have the following.
\begin{corollary} \label{cor: Euc hardness}
Let $\varepsilon>0$. If  $(1, 1-\varepsilon)$-Gap \setpT is \np-hard then 
Euclidean \dst is \np-hard to approximate within a factor of
\[
1 + \frac{\varepsilon(9\sqrt{3} -15)}{15} \geq 1 + 0.039\varepsilon.
\]
More generally, if $(a, b)$-Gap \setpT is \np-hard, then Euclidean \dst is \np-hard to approximate within a factor of
\[
\frac{b(5\sqrt{3}/9 - 1) + 1}{a(5\sqrt{3}/9 - 1) + 1}.
\]
\end{corollary}
\begin{proof}
For $p = 2$, setting $\theta_2 = 1/6$ turns out to yield the maximum hardness of approximation for embeddings of this kind (this can be verified graphically). Note that $\theta_2 = 1/6$ satisfies the requirements of Theorem \ref{thm: hardness of lp dst}: the first requirement is clear and the second follows from $3\sqrt{3}/2 + \sqrt{3}/6 = \sqrt{3}(5/3) < 3$. Then, substitute into Theorem \ref{thm: hardness of lp dst} and observe that
\[
1 < \beta_{\text{in}} + \srd/2 = \sqrt{3}/2 + \sqrt{3}/12 = 7\sqrt{3}/12,
\]
as noted in Remark \ref{rmk: role of eps and delta}. This yields hardness of approximation within a factor of 
\[
1 - \varepsilon + \frac{\varepsilon}{\sqrt{3}/2 + \sqrt{3}/18} = 1 + \frac{9\sqrt{3} -15}{15} \cdot \varepsilon,
\]
proving the first part of the result.

The second part of the proof follows almost identically. We repeat almost the same reduction as in Theorem \ref{thm: graph hardness}: The soundness case is exactly identical and the resultant Steiner tree then has cost at least 
\[
b(\sqrt{3}/2 + \sqrt{3}/18) + 1 - b = b(5\sqrt{3}/9 -1) + 1.
\]
In the completeness case, we no longer have that the sets partition the universe. Let $C$ be a maximum packing of sets. For each $\{i,j,k\} \in C$, connect $\mathbf{e_i}$, $\mathbf{e_j}$, and $\mathbf{e_k}$ to $s = \Gamma^{-1}(\{i,j,k\})$  (following the notation of Section \ref{sec:metricspace}).  Then connect $s$ to $\mathbf{0}$. For the terminals corresponding to universe elements that are not covered by the sets in $C$, we connect those terminals directly to $\mathbf{0}$. This incurs costs $a(5\sqrt{3}/9 - 1) + 1$. Taking the ratio in the two cases then yields the desired result.
\end{proof}

The below hardness of $(a,b)$-\setpT is known.
\begin{theorem}[\cite{Chlebik_Chlebikova_2006}] \label{thm: IS hardness}
We have $(0.979,0.969)$-Gap \setpT is \np-hard.
\end{theorem}
In \cite{Chlebik_Chlebikova_2006}, they prove the hardness of the Max Independent Set problem on $3$-regular graphs, a special case of \setpT. The fact that this is a special case is clear by considering the set of incident edges to each vertex as a set in the set-system and the set of all edges in the graph as the universe of elements.\footnote{The explicit bounds here comes from a combination of places in \cite{Chlebik_Chlebikova_2006}: we use Theorem 16 for the bound on $\mu_{3,k}$, the assumption at the beginning of Theorem 17 for bounding $M(H)/k$ by $\mu_{3, k} + o(1)$, the \np-hard decision problem formulation of Max Independent Set stated in the proof of Theorem 17 for our two cases, and the property $|V(H)| = 2M(H)$ used to prove Corollary 18 to apply the bound on $M(H)/k$ .}
Applying this to Corollary \ref{cor: Euc hardness} we have the following.
\begin{corollary}
Euclidean \dst is \np-hard to approximate within a factor of 1.00039.
\end{corollary}

\subsection{Hardness of Approximation Bounds for Hamming and Rectilinear \texorpdfstring{\dst}{DST}}
The general embedding described in Section \ref{subsec: ell_p embedding} notably fails for $p = 1$. However, a simple reduction from \vcov yields reasonable hardness of approximation bounds. We simultaneously get the same hardness for Hamming ($\ell_0$) \dst. 

We will make use of the following theorem.

\begin{theorem} [\cite{Chlebik_Chlebikova_2006}] \label{thm: vcov hard}
$(0.52025, 0.53036)$-\vcov is \np-hard on $4$-regular graphs.
\end{theorem}

We show the following.

\begin{theorem} \label{thm: vx cover corr}
Suppose that $(a, b)$-\vcov is \np-hard on max degree $\Delta$ graphs. Then  \dst in both the $\ell_0$ and $\ell_1$ metrics is \np-hard to approximate within a factor of
$\frac{\Delta/2 + b}{\Delta/2 + a}$.
\end{theorem}
\begin{proof}
The proof is exactly the same as the reduction from Set Cover to \dst given in Section \ref{sec:techniques}. Interpret the instance of \vcov as an instance of Set Cover with universe elements being edges and sets being the sets of incident edges to vertices. We can embed the graph used in the reduction in $\R^n$ under the $\ell_0$ and $\ell_1$ metrics by mapping the ``additional vertex'' to $\mathbf{0}$, the edges to their characteristic vectors (of Hamming weight 2), and each set of edges incident to a vertex to the characteristic vector of that vertex (of Hamming weight 1).
\end{proof}

Then, using Theorem \ref{thm: vcov hard}, we have the following result.

\begin{theorem}
We have that  \dst in both the $\ell_0$ and $\ell_1$ metrics is \np-hard to approximate within a factor of 1.004.
\end{theorem}

\section{String Metrics} \label{sec: string metrics}
In this section, we show how \apx-hardness of \dst in the Hamming metric implies \apx-hardness of \dst in the Ulam and edit distance metrics, using known near-isometric embeddings. This is Theorems \ref{thm: ulam hardness} and \ref{thm: edit hardness}, respectively. Let $\Sigma$ be an alphabet and $\Sigma^*$ be the set of strings over the alphabet $\Sigma$. Then let $\ED : \Sigma^* \times \Sigma^* \to \R$, the \textit{edit distance} between strings. That is, the minimum number of single character deletions, insertions, and substitutions required to modify one string into the other. The Ulam metric $\UD$ is the same edit distance except with input strings restricted to those without repeated characters.

\subsection{Ulam Metric}
In \cite{DBLP:journals/corr/abs-2112-03222}, they show that there exists an embedding of $n$-dimensional Hamming space into the set of permutations of $[2n]$ under the Ulam metric such that all pairwise distances are exactly scaled by a factor of $2$.

\begin{lemma}[Lemma 4.5 of \cite{DBLP:journals/corr/abs-2112-03222}]
Let $\Pi_{2n}$ denote the set of permutations of $[2n]$. There is a function $\eta: \{0,1\}^n \to \Pi_{2n}$, such that, for all $x, y \in \{0,1\}^n$, we have
\[
\ed(\eta(x), \eta(y)) = 2 \norm{x - y}_0.
\]
Moreover, for all $x \in \{0,1\}^n$, $\eta(x)$ can be computed in $O(n)$ time.
\end{lemma}

Then, given a hard instance $(P, X)$ of Hamming \dst, consider $(\eta(P), \eta(X))$ as an instance of \dst in the Ulam metric. In particular, there is a one-to-one correspondence between Steiner trees in the two instances and their costs are exactly scaled by $2$. So, hardness is exactly scaled by $2$ (using that $\eta(x)$ can be computed in $O(n)$ time). This completes the prove of Theorem \ref{thm: ulam hardness}.

\subsection{Edit Distance Metric}
In \cite{DBLP:conf/stoc/Rubinstein18}, Rubinstein shows a near-isometric embedding of Hamming space into the space of binary strings under the edit distance metric. We use an explicit formulation given in Lemma A.1 of \cite{CK19}.

\begin{lemma}[\cite{DBLP:conf/stoc/Rubinstein18}] \label{lem: rub ED embedding}
For large enough $d$, there exists a function $\eta: \{0,1\}^d \to \{0,1\}^{d'}$, where $d' = O(d \log d)$ such that for all $a, b \in \{0,1\}^d$, we have 
\[
|\ED(\eta(a), \eta(b)) - c \cdot \log d \cdot \norm{a - b}_0| = o(d'),
\]
for some constant $c$. Moreover, for any $a \in \{0,1\}^d$, $\eta(a)$ can be computed in $2^{o(d)}$ time.
\end{lemma}

First, we use the fact that Hamming \dst is \apx-hard in $O(\log n)$ dimensions. This follows from Theorem \ref{thm: dimensionality reduction} in Section \ref{sec: dimensionality reduction}. Then, given an instance $(P, X)$ of Hamming \dst in $O(\log n)$ dimensions, for $n$ large enough, consider instance $(\eta(P), \eta(X))$ of $\ED$ \dst (with $\eta$ the function from Lemma \ref{lem: rub ED embedding}). This is computable in $\poly (n)$ time by Lemma \ref{lem: rub ED embedding}. 

Now, observe that the cost of optimal Steiner trees must be $\Omega(n)$ in both the completeness and soundness cases of Hamming \dst (since the minimum distance between any two points is $1$). Now, let $T$ be a Steiner tree of $(P,X)$. By Lemma \ref{lem: rub ED embedding}, the corresponding Steiner tree of $(\eta(P), \eta(X))$ costs at most $c \cdot \log \log n \cdot \cost_0(T) + o(\log n \log \log n)$ and at least $c \cdot \log \log n \cdot \cost_0(T) - o(\log n \log \log n)$. Then, since $\cost_0(T) = \Omega(n)$ the ratio of the bounds on the costs of optimal trees in the soundness and completeness cases of the $\ED$ \dst are preserved for sufficiently large $n$ since the $\Omega(n)$ term dominates the $o(\log n \log \log n)$ term. Hence, since there is a constant gap for Hamming \dst, there is a constant gap for $\ED$ \dst. This completes the proof of Theorem \ref{thm: edit hardness}.

\section{Inapproximability in Low Dimensions} \label{sec: dimensionality reduction}

In Section \ref{subsec: ell_p embedding}, we describe gap-preserving reduction from $(\varepsilon, \delta)$-\setpT to  \dst in $n$-dimensional $\ell_p$-metric spaces. In this section, we show that this reduction can be extended to a configuration of $n$ points in an $O(\log n)$-dimensional $\ell_p$-metric space (for constant $p$) via a near-isometric embedding from \cite{CK19,CKL22}. For the dimensionality reduction for $p = \infty$, refer to \cite{CK19}. This will show that \dst in the $\ell_p$-metric is \apx-hard even in $O( \log n)$-dimensions (and essentially preserve the hardness of approximation).  In the below, let $\supp(x)$ be the set of nonzero coordinates of $x$ (the \textit{support} of $x$).

\begin{theorem}[Implicit in \cite{CK19,CKL22}] \label{thm: dimensionality reduction}
Let $\varepsilon > 0$, $B, C \in \Z^+$, and $p \in \R_{\geq 1}$. Let $P \subset \R^n$ such that $|P| = n$ and, for all $x \in P$, $\supp(x)  \subseteq [n]$ with $|\supp(x)| \leq C$ and $\norm{x}_{\infty} \leq B$. Then, there exists $\alpha = O_{C,B,\varepsilon}(1/\log n)$ and $\sigma : \R^n \to \R^{O_{C,B, \varepsilon}(\log n)}$ such that, for all $x, y \in P$,
\[
(1 - \varepsilon) \norm{x - y}_p^p \leq \alpha \norm{\sigma(x) - \sigma(y)}_p^p \leq (1 + \varepsilon)\norm{x - y}_p^p.
\]
Additionally, for all $x \in P$, $\sigma(x)$ is computable in $\poly(n)$ time.
\end{theorem}
Note that all points in instances of \dst in Section \ref{subsec: ell_p embedding} have support of size at most $3$. Hence, Theorem \ref{thm: dimensionality reduction} applies. Hence, for a fixed $p$ and any $\varepsilon > 0$, there exists an instance of \dst in the $\ell_p$-metric in $O(\log n)$ dimensions with all pairwise distances within an arbitrarily small factor of the original distances in $n$ dimensions (after scaling the configuration by $\alpha^{1/p}$). Therefore, the cost of the trees in the completeness and soundness cases change by at most that factor and we get hardness of approximation within an arbitrarily small factor of the original hardness.

To prove Theorem \ref{thm: dimensionality reduction}, we use the existence of codes from algebraic geometry.

\begin{theorem}[Existence of AG codes,  \cite{Garcia_Stichtenoth_1996, Shum_Aleshnikov_Kumar_Stichtenoth_Deolalikar_2001}] \label{thm: existence of AG codes}
Let $q$ be a prime squared at least $49$. Then, there exists a linear map $\mathcal{C} : \mathbb{F}_q^{\log_q n} \to \mathbb{F}_q^{c \log_q n}$ for some constant $c$, such that for all $x, y \in \mathbb{F}_q^{\log_q n}$ we have
\[
\| \mathcal{C}(x) - \mathcal{C}(y)\|_0 \geq (1 - 3/\sqrt{q})(c \log_q n).
\]  
\end{theorem}
To define the map $\sigma: \R^n \to \R^{O(\log n)}$ for Theorem \ref{thm: dimensionality reduction}, we first describe the map on $\{0,1\}^n_1 \subset \R^n$, the set of vectors of Hamming weight $1$ (the standard basis vectors).  Let $q$ be a squared prime at least $49$ (we will specify $q$ appropriately later). Given $i \in n$, map it bijectively to an element $x \in \mathbb{F}_q^{\log_q n}$ (by enumerating the elements of $\mathbb{F}_q^{\log _q n})$. Call this map $f_1$. Let $\mathcal{C}$ be the linear map existing from Theorem \ref{thm: existence of AG codes}. Then, apply $\mathcal{C}$ to $x$ to yield an element of  $\mathbb{F}_q^{c \log_q n}$. Finally, enumerate the elements of $\mathbb{F}_q$ from $1$ to $q$ and map $\mathcal{C}(f_1(x))$ to $\{0,1\}^{qc \log_q n}$ by interpreting  $\mathcal{C}(f_1(x))$ as $c \log_q n$ blocks of $q$ zeroes and ones, $q-1$ of which being $0$ and the $1$ indicating the number of the element from $\mathbb{F}_q$. Call this map $f_2$. For example, $(3,1,2) \in \mathbb{F}_5^3$ would be mapped to the bitstring $000100100000100 \in \{0,1\}^{15}$. We interpret this bitstring in $\{0,1\}^{qc \log_q n}$ as a vector in $\mathbb{R}^{qc \log_q n}$ by the natural inclusion map. Call this map $f_3$. Formally, $\sigma = f_3 \circ f_2 \circ \mathcal{C} \circ f_1$ for inputs in $\{0,1\}^n_1$. 
%
%
To extend $\sigma$ to the rest of $\R^n$, we define
\[
\sigma \left(\sum_{i = 1}^n a_i \e_i \right) = \sum_{j = 1}^{qc \log_q n} \begin{cases}
a_i \cdot \e_j &\text{if } |\{i \in [n] \,:\, a_i \sigma(\e_i)_j \neq 0\}| = 1, \\
0 &\text{otherwise}.
\end{cases}
\]
 Intuitively, given the large minimum distance between codewords from Theorem \ref{thm: existence of AG codes} and the bounded support of point configurations in Theorem \ref{thm: dimensionality reduction}, for most $j \in [qc \log_q n]$, $|\{i \in [n] \,:\, \sigma(\e_i)_j = 1\}| \in \{0,1\}$.

Now, let $\varepsilon, C$, $B$, $p$, and $P$ be as in Theorem \ref{thm: dimensionality reduction}. Fix $x, y \in P$. Let $x = \sum_{i \in I_1} x_i \e_i$ and  $y = \sum_{i \in I_2} y_i \e_i$, where $|I_1|, |I_2| \leq C$. Fix $i \in I_1$. Consider the set of nonzero coordinates on which $\sigma(\e_i)$ does not overlap with $\sigma(\e_j)$ for $i \neq j \in I_1 \cup I_2$. The contribution of each such coordinate $k$ to $\norm{\sigma(x) - \sigma(y)}_p^p$ is exactly $|x_i - y_i|^p$. 

By Theorem \ref{thm: existence of AG codes}, $\norm{\mathcal{C}(f_1(\e_i)) - \mathcal{C}(f_1(\e_j))}_0 \geq (1 - 3/\sqrt{q}) c \log_q n$ for each $i \neq j \in [n]$. Then, since the map $f_2$ exactly doubles Hamming distance and $f_3$ preserves Hamming distance, $\norm{\sigma(\e_i) - \sigma(\e_j)}_0 \geq (1 - 3/\sqrt{q}) 2c \log_q n$. In particular, since $\norm{\sigma(\e_i)}_0$ is exactly $c \log_q n$, this implies that $\sigma(\e_i)$ and $\sigma(\e_j)$ can overlap on at most $(3/\sqrt{q}) c\log_q n$ nonzero coordinates. Moreover, $\sigma(\e_i)$ can overlap on at most $(6C/\sqrt{q}) c\log_q n$ total nonzero coordinates with $\sigma(\e_j)$ for $i \neq j \in I$. The maximum contribution of these overlapping coordinates is $B^p$ (if, for example, $\sigma(e_i)$ overlaps with $\sigma(e_{i'})$ for $i' \in I_1$ but doesn't overlap with any mapped standard basis vector from $I_2$ and $a_i = B$). The minimum contribution is $0$. Then, we have 
\[
(1 - 6C/\sqrt{q}) c\log_q n \norm{ x - y}_p^p \leq \norm{\sigma(x) - \sigma(y)}_p^p \leq c\log_q n \norm{ x - y}_p^p + (6C/\sqrt{q})B^p c \log_q n.
\]
Setting $q$ sufficiently large then yields Theorem \ref{thm: dimensionality reduction}.

\end{document}